\documentclass[journal,10pt,doublecolumn]{IEEEtran}

\usepackage[utf8]{inputenc}
\usepackage[T1]{fontenc}

\usepackage{url}
\usepackage{cite}

\usepackage[cmex10]{amsmath}
\usepackage{amssymb}
\usepackage{amsfonts}
\usepackage{amsthm}
\usepackage{mathtools}

\usepackage{graphicx}
\usepackage{booktabs}
\usepackage{array}
\usepackage[table]{xcolor}

\usepackage[
    colorlinks=true,
    linkcolor=black,
    anchorcolor=black,
    citecolor=black,
    filecolor=black,
    menucolor=black,
    runcolor=black,
    urlcolor=black
]{hyperref}

\newtheorem{theorem}{Theorem}
\newtheorem{lemma}{Lemma}
\newtheorem{proposition}{Proposition}
\newtheorem{corollary}{Corollary}

\theoremstyle{definition}

\newtheorem{definition}{Definition}

\newcommand{\F}{\mathbb F}
\newcommand{\Fq}{\mathbb F_q}

\newcommand{\Acal}{\mathcal A}
\newcommand{\Bcal}{\mathcal B}
\newcommand{\Ucal}{\mathcal U}
\newcommand{\Vcal}{\mathcal V}
\newcommand{\Ncal}{\mathcal N}

\newcommand{\Span}{\operatorname{span}}
\newcommand{\rank}{\operatorname{rank}}
\newcommand{\one}{\mathbf 1}

\begin{document}

\title{Function Tables for Secure \\ Distributed Matrix Multiplication}

\author{
Rafael G. L. D'Oliveira,
Giulia Gaggero,
Arturo Jaramillo Gil, \\
Hiram H. L\'opez,
Cecilia Martínez-Reyes,
and Divyesh Vaghasiya%
\thanks{We thank CIMAT for its support via the Grant CICIMPI-2026-07: Workshop on Applications of Commutative Algebra and Coding Theory.}
\thanks{R. G. L. D'Oliveira is with the School of Mathematical and Statistical Sciences, Clemson University, Clemson, SC, USA. Email: rdolive@clemson.edu.}%
\thanks{G. Gaggero is with McMaster University. Email: gaggerog@mcmaster.ca.}%
\thanks{A. Jaramillo Gil is with the  Department of Mathematics, Centro de Investigaci\'on en Matem\'aticas, M\'exico. Email: jagil@cimat.mx. He was supported by Grant CBF2023-2024-2088.}%
\thanks{H. H. L\'opez is with the Department of Mathematics, Virginia Tech, VA, USA. Email: hhlopez@vt.edu. He was partially supported by the NSF grants DMS-2401558 and DMS-2502705, and the Commonwealth Cyber Initiative.}%
\thanks{C. Martínez-Reyes is with Universidad Autónoma de Zacatecas. Email: maria.reyes1@ues.edu.sv. She was supported by SECIHTI Grants 4003512 and CF-2023-G-33.}%
\thanks{D. Vaghasiya is with the University of South Florida, Tampa, FL, USA. Email: divyeshvaghasiya00@gmail.com.}%
}

\maketitle

\begin{abstract}
We introduce function tables, an entrywise representation of the coefficient functions that appear in the worker responses of a secure distributed matrix multiplication (SDMM) scheme. We work under the outer-product partition, with $K$ row blocks, $L$ column blocks, and privacy against any $T$ colluding workers, in the general model of linear encoding and linear decoding. In this representation, privacy is a rank condition on the data and mask coefficients, and decodability is linear independence of the desired entries modulo the nuisance space. Degree tables, cyclic-addition tables, and algebraic-geometry constructions are the special cases obtained by restricting the coefficient functions to a structured family; we impose no such restriction, so our converses bind every linear scheme. For $T=1$, we determine the exact optimum over every finite field $\Fq$: it is $KL+K+L$ when $q\geq3$, and $KL+K+L+1$ over $\F_2$, where the identity $z^2=z$ forces one more worker. For arbitrary $T$, we prove $N\geq KL+K+L$ and $N\geq\max\{K,L\}+T$ with no MDS hypothesis on the masks; the first is stronger than the previously known bound $KL+\max\{K,L\}+2T-1$ whenever $\min\{K,L\}\geq2T$. We then reduce field feasibility exactly to MDS existence: a scheme exists over $\Fq$ if and only if an $[\max\{K,L\}+T,T]$ linear MDS code does, and whenever it does, a Cartesian construction attains $N=(K+T)(L+T)$ over that same field. For $T=2$ this makes $q\geq\max\{K,L\}+1$ necessary and sufficient, and we give a projective-line construction with $N=KL+K+L+2$ whenever $KL+K+L$ divides $q-1$; for $K,L\geq2$ it matches the best known worker count while requiring only an element of order $KL+K+L$.
\end{abstract}


\section{Introduction}
\label{sec:introduction}

Secure distributed matrix multiplication (SDMM) allows a user with two matrices $A\in\Fq^{m\times n}$ and $B\in\Fq^{n\times p}$ over a finite field $\Fq$ to compute $AB$ with the help of $N$ workers while keeping the input matrices private. The user encodes the two matrices, sends one encoded pair to each worker, and combines the returned products to recover the desired result. The workers are honest but curious: they follow the protocol, but any set of at most $T$ workers may collude and try to learn information about the inputs. We study how many workers are required and over which fields such schemes can be constructed for given partitioning and privacy parameters.

We consider the outer-product partition. Assume, for simplicity, that $K$ divides $m$ and $L$ divides $p$. Partition
\[
A=
\begin{bmatrix}
A_1\\
A_2\\
\vdots\\
A_K
\end{bmatrix},
\qquad
B=
\begin{bmatrix}
B_1 & B_2 & \cdots & B_L
\end{bmatrix},
\]
where $A_k\in\Fq^{m/K\times n}$ and $B_\ell\in\Fq^{n\times p/L}$. Then
\[
AB=
\begin{bmatrix}
A_1B_1 & A_1B_2 & \cdots & A_1B_L\\
A_2B_1 & A_2B_2 & \cdots & A_2B_L\\
\vdots & \vdots & \ddots & \vdots\\
A_KB_1 & A_KB_2 & \cdots & A_KB_L
\end{bmatrix}.
\]
For every positive integer $r$, write $[r]=\{1,\ldots,r\}$. Computing $AB$ is therefore equivalent to computing the $KL$ block products $A_kB_\ell$, for $k\in[K]$ and $\ell\in[L]$.

We work in the general class of linear SDMM schemes formalized by Makkonen and Hollanti~\cite{10415397}. To provide privacy, the user introduces random masks $R_1,\ldots,R_T$ of the same size as the blocks $A_k$ and random masks $S_1,\ldots,S_T$ of the same size as the blocks $B_\ell$. All these masks are mutually independent, uniformly distributed over their respective matrix spaces, and independent of both $A$ and $B$. The user sends to worker $i$ two encoded matrices of the form
\begin{equation}
X_i
=
\sum_{k=1}^K a_k(i)A_k
+
\sum_{t=1}^T u_t(i)R_t,
\label{eq:X_encoding}
\end{equation}
and
\begin{equation}
Y_i
=
\sum_{\ell=1}^L b_\ell(i)B_\ell
+
\sum_{t=1}^T v_t(i)S_t,
\label{eq:Y_encoding}
\end{equation}
where $a_k(i),b_\ell(i),u_t(i),v_t(i)\in\Fq$. Worker $i$ computes $Z_i=X_iY_i$ and returns $Z_i$ to the user.

\begin{figure*}[t]
\centering

\def\FunctionTableScale{0.90}

\newcommand{\ScaledFunctionTable}[1]{%
  \makebox[\linewidth][c]{%
    \scalebox{\FunctionTableScale}{$#1$}%
  }%
}

\begingroup
\setlength{\arraycolsep}{2pt}
\renewcommand{\arraystretch}{1.15}

\begin{minipage}[t]{0.32\textwidth}
\centering
\textbf{(a) General function table}

\medskip

\ScaledFunctionTable{
\begin{array}{c|cccccc}
 & b_1 & \cdots & b_L & v_1 & \cdots & v_T\\
\hline
a_1
& \cellcolor{red!12}a_1b_1
& \cellcolor{red!12}\cdots
& \cellcolor{red!12}a_1b_L
& a_1v_1
& \cdots
& a_1v_T\\
\vdots
& \cellcolor{red!12}\vdots
& \cellcolor{red!12}\ddots
& \cellcolor{red!12}\vdots
& \vdots
& \ddots
& \vdots\\
a_K
& \cellcolor{red!12}a_Kb_1
& \cellcolor{red!12}\cdots
& \cellcolor{red!12}a_Kb_L
& a_Kv_1
& \cdots
& a_Kv_T\\
u_1
& u_1b_1
& \cdots
& u_1b_L
& u_1v_1
& \cdots
& u_1v_T\\
\vdots
& \vdots
& \ddots
& \vdots
& \vdots
& \ddots
& \vdots\\
u_T
& u_Tb_1
& \cdots
& u_Tb_L
& u_Tv_1
& \cdots
& u_Tv_T
\end{array}
}
\end{minipage}
\hfill
\begin{minipage}[t]{0.32\textwidth}
\centering
\textbf{(b) $T=1$ and $q\geq3$}

\medskip

\ScaledFunctionTable{
\begin{array}{c|ccccc}
 & z & y_1 & \cdots & y_{L-1} & \one\\
\hline
z
& \cellcolor{red!12}z^2
& \cellcolor{red!12}zy_1
& \cellcolor{red!12}\cdots
& \cellcolor{red!12}zy_{L-1}
& z\\
x_1
& \cellcolor{red!12}x_1z
& \cellcolor{red!12}x_1y_1
& \cellcolor{red!12}\cdots
& \cellcolor{red!12}x_1y_{L-1}
& x_1\\
\vdots
& \cellcolor{red!12}\vdots
& \cellcolor{red!12}\vdots
& \cellcolor{red!12}\ddots
& \cellcolor{red!12}\vdots
& \vdots\\
x_{K-1}
& \cellcolor{red!12}x_{K-1}z
& \cellcolor{red!12}x_{K-1}y_1
& \cellcolor{red!12}\cdots
& \cellcolor{red!12}x_{K-1}y_{L-1}
& x_{K-1}\\
\one
& z
& y_1
& \cdots
& y_{L-1}
& \one
\end{array}
}
\end{minipage}
\hfill
\begin{minipage}[t]{0.32\textwidth}
\centering
\textbf{(c) $T=1$ and $q=2$}

\medskip

\ScaledFunctionTable{
\begin{array}{c|ccccc}
 & y_1 & y_2 & \cdots & y_L & \one\\
\hline
x_1
& \cellcolor{red!12}x_1y_1
& \cellcolor{red!12}x_1y_2
& \cellcolor{red!12}\cdots
& \cellcolor{red!12}x_1y_L
& x_1\\
x_2
& \cellcolor{red!12}x_2y_1
& \cellcolor{red!12}x_2y_2
& \cellcolor{red!12}\cdots
& \cellcolor{red!12}x_2y_L
& x_2\\
\vdots
& \cellcolor{red!12}\vdots
& \cellcolor{red!12}\vdots
& \cellcolor{red!12}\ddots
& \cellcolor{red!12}\vdots
& \vdots\\
x_K
& \cellcolor{red!12}x_Ky_1
& \cellcolor{red!12}x_Ky_2
& \cellcolor{red!12}\cdots
& \cellcolor{red!12}x_Ky_L
& x_K\\
\one
& y_1
& y_2
& \cdots
& y_L
& \one
\end{array}
}
\end{minipage}

\endgroup

\vspace{1em}

\caption{Function tables for secure distributed matrix multiplication with $K$ row blocks of $A$ and $L$ column blocks of $B$. Row and column labels are the coefficient functions used to encode the two inputs, and each entry is their pointwise product. Shaded entries correspond to the desired block products $A_kB_\ell$; unshaded entries correspond to products involving random masks. Panel~(a) shows the general pattern with $T$ masks per input and a target of privacy against any $T$ colluding workers. Panels~(b) and~(c) show optimal constructions for privacy against any single worker: panel~(b) uses $KL+K+L$ workers over any $\Fq$ with $q\geq3$, while panel~(c) uses $KL+K+L+1$ workers over $\F_2$. Here $\one$ denotes the constant-one function. Sharing $z$ between the two encodings in panel~(b) saves one nuisance dimension. Over $\F_2$, the identity $z^2=z$ would make a desired entry coincide with a nuisance entry, so panel~(c) uses separate coordinate functions on the two sides.}
\label{fig:function_table_and_t1_constructions}
\end{figure*}

Privacy requires that the encoded matrices observed by any set of at most $T$ workers reveal no information about the inputs. More precisely, for every joint distribution of $(A,B)$ and every subset $\tau\subseteq[N]$ with $|\tau|\leq T$, we require $I(A,B;X_\tau,Y_\tau)=0$, where $X_\tau=(X_i)_{i\in\tau}$ and $Y_\tau=(Y_i)_{i\in\tau}$.

As in Makkonen and Hollanti~\cite{10415397}, we consider linear decoding. For every $k\in[K]$ and $\ell\in[L]$, there exists a function $\lambda^{(k,\ell)}:[N]\to\Fq$ such that
\[
A_kB_\ell
=
\sum_{i=1}^N\lambda^{(k,\ell)}(i)Z_i
\]
for every choice of the inputs and random masks. Throughout the paper, we refer to this property simply as decodability.

Most of the best-performing constructions known for outer-product SDMM are based on degree tables, cyclic variants, or closely related algebraic evaluation methods. In an ordinary degree table~\cite{9004505,9508383}, each data block and each random mask is assigned its own polynomial degree, and multiplying two encoded polynomials together simply adds the degrees of their terms; the table records these sums, and the number of workers the scheme needs is the number of distinct sums that appear. Cyclic constructions evaluate the same monomials at roots of unity instead of arbitrary field elements, so degrees combine cyclically rather than as ordinary sums~\cite{11195364}. Algebraic-geometry constructions replace monomials by functions from an algebraic curve, so the table records how these functions multiply rather than how degrees add~\cite{10858081}. In each case, the construction begins by choosing a structured family of coefficient functions and studying the products among them.

We introduce function tables to unify these constructions within a single framework. Every linear SDMM scheme in our model has a corresponding function table: its rows and columns are the coefficient functions used in the two encodings, and its entries are their pointwise products. Degree tables, cyclic-addition tables, and other algebraic tables are all special cases, obtained by restricting the coefficient functions to the corresponding families. Function tables are a tool for constructing general linear schemes and proving bounds.

With this framework, we characterize privacy and decodability, determine the exact minimum number of workers for $T=1$, prove general worker lower bounds, characterize exactly the fields over which a scheme in our model can exist, and give constructions with different tradeoffs between field size and worker count.

\subsection{Related Work}
\label{sec:related_work}

Early information-theoretic formulations of SDMM appeared
in~\cite{8647313,8382305}. Subsequent work considered different matrix
partitions, straggler models, and security requirements
\cite{8675905,8985291,9229375,9539194}. Other objectives include upload
and download cost~\cite{8989342,9440909} and total computation
time~\cite{9162296}. We focus on the number of worker responses needed
for decoding. In the fixed-block, one-product, full-response model
studied here, this is the number of workers used by the scheme. This
equivalence does not extend to multi-message or compressed-response
schemes.

Polynomial codes were introduced for nonsecure distributed matrix
multiplication in~\cite{yu2017polynomial} and developed further
in~\cite{8006963,8437871,8765375,8949560,yu2019lagrange,9519610,10786350}.
For SDMM under the outer-product partition, GASP and degree tables
established systematic constructions and converse bounds based on
ordinary sums of degrees~\cite{9004505,9508383}. Other direct
competitors include A3S~\cite{8675905}, root-of-unity
codes~\cite{9965858}, the single-multiplication construction
in~\cite{9523544}, and the cyclic CAT and DOG
constructions~\cite{11195364}. For $T=1$, several of these constructions
attain $KL+K+L$ over suitable fields, but prior work does not determine
the optimum over the full linear model for every field.

For $T=2$ and $K,L\geq2$, set $M=KL+K+L$. GASP uses $M+3$ workers,
whereas CATx uses $M+2$. After choosing the better orientation, A3S,
root-of-unity codes, and the single-multiplication construction
of~\cite{9523544} use $M+\min\{K,L\}+1$. The one-product specialization
of secure bivariate polynomial coding has the same count when its mask
evaluations satisfy the required rank condition~\cite{9681059}.
Algebraic-geometry constructions provide a different approach by
replacing monomials with functions on curves~\cite{10858081}. These are
the main direct comparisons for our projective-line construction.

A general framework for linear SDMM based on linear codes and their
star products was developed in~\cite{10415397}. Under MDS mask
subcodes, its converse specializes in our setting to
$N\geq KL+\max\{K,L\}+2T-1$. Function tables give an entrywise
representation of this linear model in which the desired products and
the nuisance space are explicit. Related structural results on star
products appear in~\cite{6594847,7137642,7133155}. The field-size
question is closely connected to the existence of linear MDS
codes~\cite{ball2012sets}. A product structure related to our Cartesian
construction appears in the square construction of~\cite{8647313} and
in the naive full-grid specialization of secure bivariate polynomial
coding~\cite{9681059}.

Other partitions and decoding requirements lead to different
comparisons. DFT codes, Secure MatDot, and HerA study the inner-product
partition and related field-size questions
\cite{9732990,9965839,10206764,10619357}. General grid partitions,
straggler recovery, and Byzantine robustness are considered
in~\cite{e25020266,9965858,10478018,11653890}. Precomputation can reduce
the online worker count by allowing the user to compute the mask--mask
products in advance~\cite{10619695}; this changes the decoding model,
since our decoder must recover from the worker responses alone.

Further variations include trace downloads, which reduce the amount
downloaded from each worker and may use additional workers
\cite{9606447}, private and batch matrix multiplication
\cite{8832193,8754796,9696353,9539194}, arbitrary collusion
patterns~\cite{9930803}, and Gram matrix products~\cite{10161614}.
Quantum and analog versions of SDMM have also been studied
in~\cite{nomeir2025quantum,11462243,11505934}.

\subsection{Summary of Results}
\label{sec:summary_results}

Our main results are as follows.

\begin{itemize}

\item We introduce function tables as a tool for studying linear SDMM schemes. Every linear scheme in our model has a function table, while degree tables, cyclic-addition tables, and related algebraic constructions arise as special cases.

\item For $T=1$, we determine the minimum number of workers for every $K$, $L$, and finite field $\Fq$. The optimum is $KL+K+L$ when $q\geq3$ and $KL+K+L+1$ when $q=2$.

\item For arbitrary $T$, we prove that every private and linearly decodable scheme satisfies $N\geq\max\{KL+K+L,T+\max\{K,L\}\}$. We also show that the mask coefficient matrices contain restrictions generating $[K+T,T]$ and $[L+T,T]$ linear MDS codes.

\item We characterize which fields support a scheme in our model. For fixed $K,L,T$, such a scheme exists over $\Fq$ if and only if an $[\max\{K,L\}+T,T]$ linear MDS code over $\Fq$ exists. Whenever this condition holds, the Cartesian construction gives a scheme with $N=(K+T)(L+T)$ workers. In particular, for $T=2$, a scheme exists if and only if $q\geq\max\{K,L\}+1$.

\item For $T=2$, we give a projective-line construction with $N=KL+K+L+2$ workers whenever $KL+K+L\mid(q-1)$. For $K,L\geq2$, it matches CATx in worker count while requiring an element of order $KL+K+L$, rather than $KL+K+L+2$.

\end{itemize}

\section{Main Results}
\label{sec:main_results}

We now introduce the function-table framework and state the main results proved in the later sections.

\subsection{The Function-Table Framework}
\label{sec:function_tables}

Regard the encoding coefficients in \eqref{eq:X_encoding} and \eqref{eq:Y_encoding} as functions $a_k,u_t,b_\ell,v_t:[N]\to\Fq$. For functions $f,g:[N]\to\Fq$, let $fg$ denote their pointwise product, defined by $(fg)(i)=f(i)g(i)$. When a worker multiplies its two encoded matrices, the corresponding coefficient functions multiply pointwise. We organize these products in a function table.

\begin{definition}
\label{def:function_table}
The \emph{function table} of a linear SDMM scheme has rows labeled by $a_1,\ldots,a_K,u_1,\ldots,u_T$ and columns labeled by $b_1,\ldots,b_L,v_1,\ldots,v_T$. Each entry is the pointwise product of its row and column labels.
\end{definition}

Expanding the response $Z_i=X_iY_i$ gives
\begin{align*}
Z_i
&=
\sum_{k=1}^K\sum_{\ell=1}^L
a_k(i)b_\ell(i)A_kB_\ell\\
&\quad+
\sum_{k=1}^K\sum_{t=1}^T
a_k(i)v_t(i)A_kS_t\\
&\quad+
\sum_{s=1}^T\sum_{\ell=1}^L
u_s(i)b_\ell(i)R_sB_\ell\\
&\quad+
\sum_{s=1}^T\sum_{t=1}^T
u_s(i)v_t(i)R_sS_t.
\end{align*}

The functions $a_kb_\ell$ are the coefficients of the desired products $A_kB_\ell$. Every other entry is the coefficient of a term involving at least one random mask.

Let $\Fq^{[N]}$ be the vector space over $\Fq$ of all functions from $[N]$ to $\Fq$. In this space, define $\Acal=\Span\{a_1,\ldots,a_K\}$, $\Ucal=\Span\{u_1,\ldots,u_T\}$, $\Bcal=\Span\{b_1,\ldots,b_L\}$, and  $\Vcal=\Span\{v_1,\ldots,v_T\}$. For subspaces $\mathcal X,\mathcal Y\subseteq\Fq^{[N]}$, let $\mathcal X\mathcal Y=\Span\{fg:f\in\mathcal X,\ g\in\mathcal Y\}$. The desired entries span $\Acal\Bcal$, while the nuisance entries span $\Ncal=\Acal\Vcal+\Ucal\Bcal+\Ucal\Vcal$.

\begin{theorem}
\label{thm:decodability_condition}
The scheme is linearly decodable if and only if the desired entries $a_kb_\ell$ are linearly independent modulo $\Ncal$. Equivalently, no nonzero linear combination of the desired entries belongs to $\Ncal$.
\end{theorem}

Thus, decodability can be read directly from the function table: after the nuisance space is removed, the $KL$ desired entries must remain linearly independent.

We next state the corresponding privacy condition. For $\tau\subseteq[N]$, define
\[
\begin{aligned}
\mathsf A_\tau&=(a_k(i))_{i\in\tau,\,k\in[K]},
&
\mathsf U_\tau&=(u_t(i))_{i\in\tau,\,t\in[T]},\\
\mathsf B_\tau&=(b_\ell(i))_{i\in\tau,\,\ell\in[L]},
&
\mathsf V_\tau&=(v_t(i))_{i\in\tau,\,t\in[T]}.
\end{aligned}
\]

\begin{theorem}
\label{thm:privacy_condition}
Consider a linear SDMM scheme whose random masks are mutually independent, uniformly distributed, and independent of $(A,B)$. For any fixed $\tau\subseteq[N]$, the condition $I(A,B;X_\tau,Y_\tau)=0$ holds for every joint distribution of $(A,B)$ if and only if $\rank[\mathsf A_\tau\ \mathsf U_\tau]
=
\rank\mathsf U_\tau$ and $\rank[\mathsf B_\tau\ \mathsf V_\tau]
=
\rank\mathsf V_\tau$.
\end{theorem}

The rank conditions say that every linear combination of the observations in $\tau$ that cancels the masks must also cancel the data.

\begin{definition}[MDS mask condition]
\label{def:mds_mask_condition}
The mask functions $u_1,\ldots,u_T$ satisfy the \emph{$T$-MDS mask condition} if $\rank\mathsf U_\tau=|\tau|$ for every $\tau\subseteq[N]$ with $|\tau|\leq T$. The condition for $v_1,\ldots,v_T$ is defined analogously.
\end{definition}

\begin{corollary}
\label{cor:mds_privacy}
If the mask functions on both sides satisfy the $T$-MDS mask condition, then the scheme is $T$-private.
\end{corollary}

Several constructions from the literature arise by restricting the coefficient functions to families with simple multiplication rules. Evaluated univariate monomials give ordinary degree tables, while evaluation on a multiplicative subgroup gives cyclic-addition tables. Product evaluation sets give Cartesian tables, and evaluations of functions on algebraic curves give algebraic-geometry tables described by star products. Thus, the same privacy and decodability conditions apply to degree-table, cyclic, Cartesian, and algebraic-geometry constructions~\cite{9508383,11195364,9681059,10415397,10858081,10206764}.

\subsection[Optimal Schemes for T=1]{Optimal Schemes for $T=1$}
\label{sec:t1_statement}

Our first application of the framework resolves the case $T=1$ over every finite field.

\begin{theorem}
\label{thm:t1_optimality}
For $T=1$, the minimum number of workers among all private and linearly decodable schemes over $\Fq$ is $KL+K+L$ when $q\geq3$ and $KL+K+L+1$ when $q=2$.
\end{theorem}

Over a nonbinary field, the two sides can share one nonconstant function $z$ while $z^2$ remains independent of the nuisance space. This saves one nuisance dimension. Over $\F_2$, every function satisfies $z^2=z$, so this saving is impossible and one additional nuisance dimension is necessary. The constructions are illustrated in Section~\ref{sec:t1_examples}, and the matching constructions and lower bounds are proved in Section~\ref{sec:t1_constructions}.

\subsection{General Bounds and Field Feasibility}
\label{sec:general_statements}

We next give general bounds and characterize the fields over which a scheme can exist.

\begin{theorem}
\label{thm:general_lower_bound}
Every $T$-private and linearly decodable scheme over $\Fq$ satisfies $N\geq KL+K+L$.
\end{theorem}

The $KL$ desired entries must remain linearly independent modulo nuisance, while the nuisance space has dimension at least $K+L$. Together, they require at least $KL+K+L$ dimensions in the space of worker responses.

Privacy and decodability also force MDS structure in the mask coefficient matrices.

\begin{theorem}
\label{thm:necessary_mds_restrictions}
Every $T$-private and linearly decodable scheme satisfies $\rank\mathsf U_{[N]}=\rank\mathsf V_{[N]}=T$. Moreover, there are subsets $\tau_A,\tau_B\subseteq[N]$ with $|\tau_A|=K+T$ and $|\tau_B|=L+T$ such that $\mathsf U_{\tau_A}^{\mathsf T}$ generates a $[K+T,T]$ linear MDS code and $\mathsf V_{\tau_B}^{\mathsf T}$ generates an $[L+T,T]$ linear MDS code.
\end{theorem}

The two MDS restrictions may involve different sets of workers, and the complete mask matrices need not be MDS. Their lengths give the following additional bound.

\begin{corollary}
\label{cor:combined_worker_lower_bound}
Every $T$-private and linearly decodable scheme satisfies $N\geq
\max\{KL+K+L,T+\max\{K,L\}\}$.
\end{corollary}

The MDS restrictions also characterize field feasibility.

\begin{theorem}
\label{thm:field_feasibility}
A $T$-private and linearly decodable scheme exists over $\Fq$ if and only if both a $[K+T,T]$ linear MDS code and an $[L+T,T]$ linear MDS code exist over $\Fq$. Equivalently, such a scheme exists if and only if a $[\max\{K,L\}+T,T]$ linear MDS code exists over $\Fq$. Whenever these conditions hold, there is a scheme with $N=(K+T)(L+T)$ workers.
\end{theorem}

Necessity follows from Theorem~\ref{thm:necessary_mds_restrictions}. Conversely, generators of the required MDS codes give a Cartesian construction with $(K+T)(L+T)$ workers. Thus, the Cartesian construction works over every field on which a scheme can exist.

Standard bounds and constructions for linear MDS codes give the following consequences.

\begin{corollary}
\label{cor:field_size_consequences}
If a $T$-private and linearly decodable scheme exists over $\Fq$ with $T\geq2$, then $q\geq\max\{K,L\}+1$. If also $\max\{K,L\}\geq2$, then $q\geq T+1$. Conversely, a scheme exists whenever $q\geq\max\{K,L\}+T-1$.
\end{corollary}

For $T=2$, the necessary field-size condition is also sufficient.

\begin{corollary}
\label{cor:t2_field_feasibility}
For $T=2$, a private and linearly decodable scheme exists over $\Fq$ if and only if $q\geq\max\{K,L\}+1$.
\end{corollary}

The criterion is automatic in two boundary cases. When $T=1$, the required $[\max\{K,L\}+1,1]$ MDS code is a repetition code and exists over every field. When $K=L=1$, the required $[T+1,T]$ MDS code is a single-parity-check code and also exists over every field.

Theorem~\ref{thm:field_feasibility} requires MDS structure only on suitable restrictions of the mask matrices. Requiring either complete mask matrix to be MDS imposes a stronger field-size condition.

\begin{proposition}
\label{prop:mds_field_size_bound}
Let $T\geq2$. If either $u_1,\ldots,u_T$ or $v_1,\ldots,v_T$ satisfies the $T$-MDS mask condition, $q\geq N-T+1$.
\end{proposition}

In particular, for $T=2$, MDS masks on the complete worker set require $q\geq N-1$. The Cartesian construction avoids this restriction because its mask coefficient rows repeat across the worker set.

\subsection[The Projective-Line Construction for T=2]{The Projective-Line Construction for $T=2$}
\label{sec:projective_statement}

The freedom to choose other coefficient functions leads to a new construction for $T=2$.

\begin{figure*}[t]
\centering

\def\PanelScale{0.88}

\begin{minipage}[t]{0.25\textwidth}
\vspace{0pt}
\centering

{\small\textbf{(a) $\F_3$ table}\par}

\medskip

\begingroup
\setlength{\arraycolsep}{5pt}
\renewcommand{\arraystretch}{1.2}
\scalebox{\PanelScale}{$
\begin{array}{c|ccc}
 & z & y & \one\\
\hline
z
& \cellcolor{red!12}z^2
& \cellcolor{red!12}zy
& z\\
x
& \cellcolor{red!12}xz
& \cellcolor{red!12}xy
& x\\
\one
& z
& y
& \one
\end{array}
$}
\endgroup

\vspace{1.5em}

{\small\textbf{(b) $\F_2$ table}\par}

\medskip

\begingroup
\setlength{\arraycolsep}{5pt}
\renewcommand{\arraystretch}{1.2}
\scalebox{\PanelScale}{$
\begin{array}{c|ccc}
 & y_1 & y_2 & \one\\
\hline
x_1
& \cellcolor{red!12}x_1y_1
& \cellcolor{red!12}x_1y_2
& x_1\\
x_2
& \cellcolor{red!12}x_2y_1
& \cellcolor{red!12}x_2y_2
& x_2\\
\one
& y_1
& y_2
& \one
\end{array}
$}
\endgroup

\end{minipage}
\hfill
\begin{minipage}[t]{0.35\textwidth}
\vspace{0pt}
\centering

{\small\textbf{(c) $\F_3$ scheme}\par}

\smallskip

{\scriptsize $p_i=(z_i,x_i,y_i)$\par}

\medskip

\begingroup
\setlength{\arraycolsep}{3pt}
\renewcommand{\arraystretch}{1.16}
\scalebox{\PanelScale}{$
\begin{array}{c|ccc}
\text{Worker }i & p_i & X_i & Y_i\\
\hline
1 &(0,0,0) &R_1&S_1\\
2 &(0,1,0) &A_2+R_1&S_1\\
3 &(0,0,1) &R_1&B_2+S_1\\
4 &(1,0,0) &A_1+R_1&B_1+S_1\\
5 &(2,0,0) &2A_1+R_1&2B_1+S_1\\
6 &(1,1,0) &A_1+A_2+R_1&B_1+S_1\\
7 &(1,0,1) &A_1+R_1&B_1+B_2+S_1\\
8 &(0,1,1) &A_2+R_1&B_2+S_1
\end{array}
$}
\endgroup

\end{minipage}
\hfill
\begin{minipage}[t]{0.35\textwidth}
\vspace{0pt}
\centering

{\small\textbf{(d) $\F_2$ scheme}\par}

\smallskip

{\scriptsize $p_i=(x_{1,i},x_{2,i},y_{1,i},y_{2,i})$\par}

\medskip

\begingroup
\setlength{\arraycolsep}{3pt}
\renewcommand{\arraystretch}{1.16}
\scalebox{\PanelScale}{$
\begin{array}{c|ccc}
\text{Worker }i & p_i & X_i & Y_i\\
\hline
1 &(0,0,0,0) &R_1&S_1\\
2 &(1,0,0,0) &A_1+R_1&S_1\\
3 &(0,1,0,0) &A_2+R_1&S_1\\
4 &(0,0,1,0) &R_1&B_1+S_1\\
5 &(0,0,0,1) &R_1&B_2+S_1\\
6 &(1,0,1,0) &A_1+R_1&B_1+S_1\\
7 &(1,0,0,1) &A_1+R_1&B_2+S_1\\
8 &(0,1,1,0) &A_2+R_1&B_1+S_1\\
9 &(0,1,0,1) &A_2+R_1&B_2+S_1
\end{array}
$}
\endgroup

\end{minipage}

\vspace{1em}

\caption{Optimal schemes for computing $AB$ when $A$ is divided into two row blocks and $B$ into two column blocks ($K=L=2$), while keeping both inputs private from any single worker ($T=1$). Panels~(a) and~(b) show the function tables over $\F_3$ and $\F_2$, respectively. Each entry is the pointwise product of its row and column coefficient functions. Shaded entries correspond to the four desired products $A_kB_\ell$, while unshaded entries correspond to products involving random masks. Panels~(c) and~(d) give the evaluation points and explicit encodings for eight workers over $\F_3$ and nine workers over $\F_2$, respectively; both worker counts are optimal. Here $R_1$ and $S_1$ are independent uniform random masks, and $\one$ denotes the constant-one function. Each row specifies one worker, which receives $(X_i,Y_i)$ and returns $Z_i=X_iY_i$. The chosen evaluation points make the distinct table functions linearly independent, allowing the user to recover all four desired products by linear combinations of the responses.}
\label{fig:t1_small_field_examples}
\end{figure*}

\begin{theorem}[Projective-line construction]
\label{thm:projective_line_construction}
If $KL+K+L$ divides $q-1$, then there exists a private and linearly decodable scheme with $T=2$ over $\Fq$ using $N=KL+K+L+2$ workers.
\end{theorem}

The construction evaluates homogeneous monomials on a multiplicative subgroup of order $KL+K+L$, together with the two additional projective points $0$ and $\infty$. The subgroup creates the nuisance collisions needed to reduce the worker count, while the two additional points separate the desired entries from the two nuisance entries with which they would otherwise collide. The mask functions satisfy the $2$-MDS mask condition.

Combining the construction with Theorem~\ref{thm:general_lower_bound} determines the optimum to within two workers.

\begin{corollary}
\label{cor:t2_two_worker_gap}
If $KL+K+L$ divides $q-1$, then the minimum number of workers among all private and linearly decodable schemes with $T=2$ is at least $KL+K+L$ and at most $KL+K+L+2$.
\end{corollary}

The construction also attains the field-size bound for schemes with MDS masks.

\begin{corollary}[Tightness of the MDS-mask bound]
\label{cor:projective_field_optimality}
Suppose that $KL+K+L+1$ is a prime power. Then there exists a private and linearly decodable scheme with $T=2$ over $\F_{KL+K+L+1}$ using $N=KL+K+L+2$ workers. It satisfies $q=N-1$ and therefore meets the bound in Proposition~\ref{prop:mds_field_size_bound} with equality.
\end{corollary}

For $K,L\geq2$, CATx also uses $KL+K+L+2$ workers but requires a primitive $(KL+K+L+2)$-th root of unity, whereas the projective-line construction requires an element of order $KL+K+L$~\cite{11195364}. The projective-line construction therefore achieves the same worker count over fields on which CATx cannot be realized.

\section{Examples}
\label{sec:examples}

We illustrate the framework with optimal $T=1$ schemes over $\F_3$ and $\F_2$ and with two constructions for $K=L=3$ and $T=2$. In each example, we specify the coefficient functions and worker encodings and then verify decodability and privacy.

\begin{figure*}[t]
\centering

\def\ComparisonTableScale{0.90}

\begin{minipage}[t]{0.32\textwidth}
\vspace{0pt}
\centering

{\small\textbf{(a) Cartesian}\par}

\smallskip

{\scriptsize Tensor products; $N=25$ over $\F_4$\par}

\medskip

\begingroup
\setlength{\arraycolsep}{2.5pt}
\renewcommand{\arraystretch}{1.18}
\scalebox{\ComparisonTableScale}{$
\begin{array}{c|ccc|cc}
\otimes
& e_1
& e_2
& e_3
& \cellcolor{blue!15}r
& \cellcolor{blue!15}s\\
\hline
e_1
& \cellcolor{red!12}e_1\otimes e_1
& \cellcolor{red!12}e_1\otimes e_2
& \cellcolor{red!12}e_1\otimes e_3
& e_1\otimes r
& e_1\otimes s\\
e_2
& \cellcolor{red!12}e_2\otimes e_1
& \cellcolor{red!12}e_2\otimes e_2
& \cellcolor{red!12}e_2\otimes e_3
& e_2\otimes r
& e_2\otimes s\\
e_3
& \cellcolor{red!12}e_3\otimes e_1
& \cellcolor{red!12}e_3\otimes e_2
& \cellcolor{red!12}e_3\otimes e_3
& e_3\otimes r
& e_3\otimes s\\
\hline
\cellcolor{green!15}r
& r\otimes e_1
& r\otimes e_2
& r\otimes e_3
& r\otimes r
& r\otimes s\\
\cellcolor{green!15}s
& s\otimes e_1
& s\otimes e_2
& s\otimes e_3
& s\otimes r
& s\otimes s
\end{array}
$}
\endgroup

\end{minipage}
\hfill
\begin{minipage}[t]{0.215\textwidth}
\vspace{0pt}
\centering

{\small\textbf{(b) $\mathsf{GASP}_1$}\par}

\smallskip

{\scriptsize Ordinary addition; $N=18$\par}

\medskip

\begingroup
\setlength{\arraycolsep}{3.2pt}
\renewcommand{\arraystretch}{1.18}
\scalebox{\ComparisonTableScale}{$
\begin{array}{c|ccc|cc}
+
& 0 & 3 & 6
& \cellcolor{blue!15}9
& \cellcolor{blue!15}10\\
\hline
0
& \cellcolor{red!12}0
& \cellcolor{red!12}3
& \cellcolor{red!12}6
& 9 & 10\\
1
& \cellcolor{red!12}1
& \cellcolor{red!12}4
& \cellcolor{red!12}7
& 10 & 11\\
2
& \cellcolor{red!12}2
& \cellcolor{red!12}5
& \cellcolor{red!12}8
& 11 & 12\\
\hline
\cellcolor{green!15}9
& 9 & 12 & 15 & 18 & 19\\
\cellcolor{green!15}12
& 12 & 15 & 18 & 21 & 22
\end{array}
$}
\endgroup

\end{minipage}
\hfill
\begin{minipage}[t]{0.215\textwidth}
\vspace{0pt}
\centering

{\small\textbf{(c) $\mathsf{CAT}_x$}\par}

\smallskip

{\scriptsize Modulo $17$; $N=17$\par}

\medskip

\begingroup
\setlength{\arraycolsep}{3.2pt}
\renewcommand{\arraystretch}{1.18}
\scalebox{\ComparisonTableScale}{$
\begin{array}{c|ccc|cc}
\oplus_{17}
& 0 & 1 & 2
& \cellcolor{blue!15}16
& \cellcolor{blue!15}3\\
\hline
0
& \cellcolor{red!12}0
& \cellcolor{red!12}1
& \cellcolor{red!12}2
& 16 & 3\\
4
& \cellcolor{red!12}4
& \cellcolor{red!12}5
& \cellcolor{red!12}6
& 3 & 7\\
8
& \cellcolor{red!12}8
& \cellcolor{red!12}9
& \cellcolor{red!12}10
& 7 & 11\\
\hline
\cellcolor{green!15}12
& 12 & 13 & 14 & 11 & 15\\
\cellcolor{green!15}13
& 13 & 14 & 15 & 12 & 16
\end{array}
$}
\endgroup

\end{minipage}
\hfill
\begin{minipage}[t]{0.215\textwidth}
\vspace{0pt}
\centering

{\small\textbf{(d) Projective line}\par}

\smallskip

{\scriptsize Addition under $\boxplus$; $N=17$\par}

\medskip

\begingroup
\setlength{\arraycolsep}{3.2pt}
\renewcommand{\arraystretch}{1.18}
\scalebox{\ComparisonTableScale}{$
\begin{array}{c|ccc|cc}
\boxplus
& 4 & 8 & 12
& \cellcolor{blue!15}0
& \cellcolor{blue!15}16\\
\hline
1
& \cellcolor{red!12}5
& \cellcolor{red!12}9
& \cellcolor{red!12}13
& 1 & 2\\
2
& \cellcolor{red!12}6
& \cellcolor{red!12}10
& \cellcolor{red!12}14
& 2 & 3\\
3
& \cellcolor{red!12}7
& \cellcolor{red!12}11
& \cellcolor{red!12}15
& 3 & 4\\
\hline
\cellcolor{green!15}0
& 4 & 8 & 12
& \cellcolor{yellow!18}0
& 1\\
\cellcolor{green!15}4
& 8 & 12 & 1
& 4
& \cellcolor{yellow!18}\infty
\end{array}
$}
\endgroup

\end{minipage}

\vspace{1em}

\caption{Comparison of four schemes for multiplying $A$ and $B$, with $K=L=3$ and $T=2$. Panel~(a) shows the Cartesian construction, which uses tensor products and $25$ workers over every finite field $\F_q$ with $q\geq4$, including $\F_4$, the smallest possible field in our model. Panel~(b) shows $\mathsf{GASP}_1$, an addition table of integer exponents with $18$ workers. It works, for example, over $\F_{27}$ and $\F_{29}$, and can be realized over every finite field with $q\geq53$. Panel~(c) shows $\mathsf{CAT}_x$, which uses addition modulo $17$ and $17$ workers. Its root-of-unity realization requires $17\mid(q-1)$, making $\F_{103}$ the smallest suitable field. Panel~(d) shows the projective-line construction, which reduces exponent sums modulo $15$ to labels $1,\ldots,15$, except that the sums $0$ and $20$ receive the separate labels $0$ and $\infty$ shown in yellow. This construction uses $17$ workers whenever $15\mid(q-1)$, with $\F_{16}$ the smallest suitable field. The workers evaluate at the fifteen elements of a multiplicative subgroup together with the two additional points $0$ and $\infty$. These additional points distinguish the yellow mask--mask entries from desired entries that agree with them on the subgroup.}
\label{fig:t2_table_comparison}
\end{figure*}

\subsection{Optimal $T=1$ Schemes over $\F_3$ and $\F_2$}
\label{sec:t1_examples}

Take $K=L=2$ and $T=1$. Thus, $A$ is partitioned into $A_1,A_2$, while $B$ is partitioned into $B_1,B_2$. The four products to be recovered are $A_1B_1,A_1B_2,A_2B_1$, and $A_2B_2$.

In the constructions below, we choose row and column functions on a finite set $\Omega$ and then select one evaluation point $p_i$ for each worker $i\in[N]$. A function $f:\Omega\to\Fq$ gives the coefficient function $i\mapsto f(p_i)$ on $[N]$. When the domain is clear, we use the same symbol for the original function and its evaluations. We use this procedure to construct an optimal scheme over $\F_3$ with eight workers and an optimal scheme over $\F_2$ with nine workers.

For the construction over $\F_3$, consider points $(\alpha,\beta,\gamma)\in\F_3^3$ and define $z(\alpha,\beta,\gamma)=\alpha$, $x(\alpha,\beta,\gamma)=\beta$, and $y(\alpha,\beta,\gamma)=\gamma$. We choose the row functions $z,x,\one$ and the column functions $z,y,\one$. If $p_i=(z_i,x_i,y_i)$ is assigned to worker $i$, set $a_1(i)=b_1(i)=z(p_i)$, $a_2(i)=x(p_i)$, $b_2(i)=y(p_i)$, and $u_1(i)=v_1(i)=1$. The encodings are
\[
X_i=z_iA_1+x_iA_2+R_1,
\qquad
Y_i=z_iB_1+y_iB_2+S_1.
\]

For the binary construction, consider points $(\alpha_1,\alpha_2,\beta_1,\beta_2)\in\F_2^4$ and use the four coordinate functions $x_1,x_2,y_1,y_2$. We choose the row functions $x_1,x_2,\one$ and the column functions $y_1,y_2,\one$. If $p_i=(x_{1,i},x_{2,i},y_{1,i},y_{2,i})$ is assigned to worker $i$, then
\[
X_i=x_{1,i}A_1+x_{2,i}A_2+R_1,
\qquad
Y_i=y_{1,i}B_1+y_{2,i}B_2+S_1.
\]

Figure~\ref{fig:t1_small_field_examples} shows the two function tables, the chosen evaluation points, and the resulting encodings. In both constructions, worker $i$ returns $Z_i=X_iY_i$.

For the construction over $\F_3$, the function table contains the eight functions $\one,z,x,y,z^2,zy,xz,xy$. The points in Figure~\ref{fig:t1_small_field_examples}(c) make these functions linearly independent. In a linear relation among them, the first three points force the coefficients of $\one,x,y$ to vanish, the fourth and fifth points force the coefficients of $z,z^2$ to vanish, and the final three points force the coefficients of $xz,zy,xy$ to vanish. The desired entries $z^2,zy,xz,xy$ are therefore linearly independent modulo $\Span\{\one,z,x,y\}$, so the scheme is linearly decodable.

For the binary construction, the function table contains the nine functions $\one,x_1,x_2,y_1,y_2,x_1y_1,x_1y_2,x_2y_1,x_2y_2$. The first five points in Figure~\ref{fig:t1_small_field_examples}(d) separate the constant and coordinate functions, while each of the final four points separates one product $x_ky_\ell$. The desired entries are therefore linearly independent modulo $\Span\{\one,x_1,x_2,y_1,y_2\}$, and the scheme is linearly decodable.

In both constructions, the mask functions are constant and nonzero. Each worker therefore receives independently masked encodings of $A$ and $B$, giving $T=1$ privacy.

The two examples attain the values in Theorem~\ref{thm:t1_optimality}: eight workers over $\F_3$ and nine workers over $\F_2$. They also show concretely why the nonbinary construction can save one worker by sharing the function $z$, while the identity $z^2=z$ prevents the same saving over $\F_2$.

\subsection{The Cartesian Construction over $\F_4$}
\label{sec:cartesian_example}

We next take $K=L=3$ and $T=2$. Corollary~\ref{cor:t2_field_feasibility} requires $q\geq4$, so $\F_4$ is the smallest field over which a scheme with these parameters can exist. We construct a scheme over $\F_4$ using $25$ workers.

Figure~\ref{fig:t2_table_comparison} shows the Cartesian and projective-line function tables, together with $\mathsf{GASP}_1$ and $\mathsf{CAT}_x$ for context~\cite{9004505,11195364}.

Write $\F_4=\{0,1,\omega,\omega^2\}$, where $\omega^2+\omega+1=0$. Define five functions from $\{1,\ldots,5\}$ to $\F_4$ through their values:
\[
\begin{array}{c|ccccc}
i & 1 & 2 & 3 & 4 & 5\\
\hline
e_1(i) & 1 & 0 & 0 & 0 & 0\\
e_2(i) & 0 & 1 & 0 & 0 & 0\\
e_3(i) & 0 & 0 & 1 & 0 & 0\\
r(i)   & 1 & 1 & 1 & 1 & 0\\
s(i)   & 0 & 1 & \omega & \omega^2 & 1
\end{array}.
\]
We use $e_1,e_2,e_3$ as the data functions and $r,s$ as the mask functions on both sides of the table. These five functions form a basis of all functions from $[5]$ to $\F_4$.

The workers are indexed by pairs $(i,j)\in[5]\times[5]$. Row functions are evaluated at $i$, and column functions are evaluated at $j$. For functions $f,g:[5]\to\F_4$, define $(f\otimes g)(i,j)=f(i)g(j)$. Thus, the entry in row $f$ and column $g$ is $f\otimes g$. The row and column labels in Figure~\ref{fig:t2_table_comparison}(a) are both $e_1,e_2,e_3,r,s$, and the nine desired entries are $e_k\otimes e_\ell$ for $k,\ell\in[3]$.

Evaluating the row functions at $i$ and the column functions at $j$ gives
\[
X_i=e_1(i)A_1+e_2(i)A_2+e_3(i)A_3+r(i)R_1+s(i)R_2
\]
and
\[
Y_j=e_1(j)B_1+e_2(j)B_2+e_3(j)B_3+r(j)S_1+s(j)S_2.
\]
The five possible encodings on each side are
\[
\begin{aligned}
X_1&=A_1+R_1,
&
Y_1&=B_1+S_1,
\\
X_2&=A_2+R_1+R_2,
&
Y_2&=B_2+S_1+S_2,
\\
X_3&=A_3+R_1+\omega R_2,
&
Y_3&=B_3+S_1+\omega S_2,
\\
X_4&=R_1+\omega^2R_2,
&
Y_4&=S_1+\omega^2S_2,
\\
X_5&=R_2,
&
Y_5&=S_2.
\end{aligned}
\]
Worker $(i,j)$ receives $(X_i,Y_j)$ and returns $Z_{ij}=X_iY_j$. Since every pair $(i,j)\in[5]\times[5]$ is used, the construction has $25$ workers.

Since $e_1,e_2,e_3,r,s$ form a basis of the functions on $[5]$, their $25$ tensor products form a basis of the functions on $[5]\times[5]$. Hence all $25$ entries of the Cartesian function table are linearly independent. In particular, the nine desired entries are linearly independent modulo the nuisance entries, so the scheme is linearly decodable.

The five mask-coefficient rows are
\[
(r(i),s(i))
\in
\left\{
(1,0),\,
(1,1),\,
(1,\omega),\,
(1,\omega^2),\,
(0,1)
\right\},
\]
and any two distinct rows are linearly independent. Consider two workers $(i,j)$ and $(i',j')$. If $i\neq i'$, their $A$-side mask rows are independent. If $i=i'$, they receive the same $A$-side encoding, so the second worker gives no additional observation of $A$. The same argument applies to the $B$-side using $j$ and $j'$. Theorem~\ref{thm:privacy_condition} therefore gives privacy against any two workers.

Across the $25$ workers, each $A$-side mask row is repeated five times, once for every value of $j$, and each $B$-side mask row is repeated five times, once for every value of $i$. The complete mask matrices are therefore not MDS. This repetition allows the construction to use $25$ workers over $\F_4$; imposing the $2$-MDS mask condition across all workers would require $q\geq24$.

A related rectangular-grid construction also uses $25$ workers but requires five distinct nonzero field elements on each side, whereas the Cartesian construction operates over $\F_4$, the smallest field on which any scheme with these parameters can exist~\cite{9681059}.

\subsection{The Projective-Line Construction over $\F_{16}$}
\label{sec:projective_example}

We continue with $K=L=3$ and $T=2$. The Cartesian construction operates over the smallest possible field but uses $25$ workers. We now choose a different function table that gives a $17$-worker scheme over $\F_{16}$.

For $0\leq a\leq4$, let $f_a(x,y)=x^ay^{4-a}$, and for $0\leq b\leq16$, let $g_b(x,y)=x^by^{16-b}$. Their products are homogeneous monomials of degree $20$. Writing $h_e(x,y)=x^ey^{20-e}$ for $0\leq e\leq20$, we have $f_ag_b=h_{a+b}$.

On the $A$-side, we use $f_1,f_2,f_3$ as the data functions and $f_0,f_4$ as the mask functions. On the $B$-side, we use $g_4,g_8,g_{12}$ as the data functions and $g_0,g_{16}$ as the mask functions. Thus, before evaluation, the exponent of each table entry is the sum of its row and column labels.

We evaluate the functions at the $17$ points of $\mathbb P^1(\F_{16})$: the points $(t,1)$ for $t\in\F_{16}^{\times}$ and the two endpoints $(0,1)$ and $(1,0)$.

At the fifteen points $(t,1)$, we have $h_e(t,1)=t^e$. Since $t^{15}=1$, exponents that differ by $15$ give the same values at these points. As functions on the evaluation set, $h_{16}=h_1$, $h_{17}=h_2$, $h_{18}=h_3$, and $h_{19}=h_4$. These identities also hold at the endpoints because all the functions involved vanish there.

The exponents $0$ and $20$ must be kept separate. The function $h_0$ is nonzero at $(0,1)$, while $h_{20}$ is nonzero at $(1,0)$. For the row and column labels above, define
\[
r\boxplus c
=
\begin{cases}
0, & r+c=0,\\
\infty, & r+c=20,\\
\langle r+c\rangle_{15}, & \text{otherwise},
\end{cases}
\]
where $\langle e\rangle_{15}$ is the unique element of $[15]$ congruent to $e$ modulo $15$. The label $j\in[15]$ represents $h_j$, the label $0$ represents $h_0$, and the label $\infty$ represents $h_{20}$.

This gives the projective addition table in Figure~\ref{fig:t2_table_comparison}(d). Its row labels are $1,2,3,0,4$, corresponding to $f_1,f_2,f_3,f_0,f_4$, and its column labels are $4,8,12,0,16$, corresponding to $g_4,g_8,g_{12},g_0,g_{16}$.

At a point $(x,y)\in\mathbb P^1(\F_{16})$, define
\begin{align*}
X(x,y)
&=f_1(x,y)A_1+f_2(x,y)A_2+f_3(x,y)A_3\\
&\quad+f_0(x,y)R_1+f_4(x,y)R_2,\\[0.4em]
Y(x,y)
&=g_4(x,y)B_1+g_8(x,y)B_2+g_{12}(x,y)B_3\\
&\quad+g_0(x,y)S_1+g_{16}(x,y)S_2.
\end{align*}
The worker assigned to $(x,y)$ receives $(X(x,y),Y(x,y))$ and returns $Z(x,y)=X(x,y)Y(x,y)$.

At a point $(t,1)$ with $t\in\F_{16}^{\times}$, the encodings are
\begin{align*}
X(t,1)
&=
tA_1+t^2A_2+t^3A_3+R_1+t^4R_2,\\
Y(t,1)
&=
t^4B_1+t^8B_2+t^{12}B_3+S_1+tS_2,
\end{align*}
where we used $t^{16}=t$. At the endpoints, the encoded pairs are $(X(0,1),Y(0,1))=(R_1,S_1)$ and $(X(1,0),Y(1,0))=(R_2,S_2)$. This specifies all $17$ worker assignments.

The functions $h_0,h_1,\ldots,h_{15},h_{20}$ are linearly independent on the evaluation set. Indeed, in a linear relation among them, evaluation at $(0,1)$ forces the coefficient of $h_0$ to vanish, and evaluation at $(1,0)$ forces the coefficient of $h_{20}$ to vanish. The remaining relation has the form $\sum_{e=1}^{15}c_et^e=0$ for every $t\in\F_{16}^{\times}$. Dividing by $t$ gives a polynomial of degree at most $14$ that vanishes at all fifteen nonzero elements of $\F_{16}$, so every coefficient is zero.

The nine desired entries in Figure~\ref{fig:t2_table_comparison}(d) have labels $5,6,7,9,10,11,13,14,15$. The nuisance entries are spanned by the functions with labels $0,\infty,1,2,3,4,8,12$. Together, the desired and nuisance labels are exactly the seventeen independent labels $0,1,\ldots,15,\infty$. The desired functions are therefore linearly independent modulo the nuisance space, and the scheme is linearly decodable.

At a point $(t,1)$, the $A$-side mask coefficients are $(f_0(t,1),f_4(t,1))=(1,t^4)$. Since $\gcd(4,15)=1$, the map $t\mapsto t^4$ permutes $\F_{16}^{\times}$, so these fifteen mask rows are distinct. The endpoints give the additional rows $(1,0)$ and $(0,1)$. Any two of the resulting seventeen rows are linearly independent.

Similarly, the $B$-side mask coefficients at $(t,1)$ are $(g_0(t,1),g_{16}(t,1))=(1,t)$, and the endpoints again give $(1,0)$ and $(0,1)$. Hence the $B$-side mask matrix also has every pair of rows linearly independent. Both sides satisfy the $2$-MDS mask condition, so the scheme is private against any two workers.

Thus, the projective-line construction gives a private and linearly decodable scheme with $17$ workers over $\F_{16}$. The $\mathsf{CAT}_x$ construction also uses $17$ workers for these parameters but requires a primitive seventeenth root of unity, so its smallest admissible field is $\F_{103}$~\cite{11195364}.

\section{Decodability and Privacy}
\label{sec:function_table_proofs}

We now prove the decodability and privacy conditions stated in Section~\ref{sec:main_results} and record the dimension bound used below.

\subsection{Decodability}

For $\lambda,f:[N]\to\Fq$, write $\langle\lambda,f\rangle=\sum_{i=1}^N\lambda(i)f(i)$. Expanding a linear combination of the worker responses gives
\begin{align}
\sum_{i=1}^N\lambda(i)Z_i
&=
\sum_{k=1}^K\sum_{\ell=1}^L
\langle\lambda,a_kb_\ell\rangle A_kB_\ell
\notag\\
&\quad+
\sum_{k=1}^K\sum_{t=1}^T
\langle\lambda,a_kv_t\rangle A_kS_t
\notag\\
&\quad+
\sum_{s=1}^T\sum_{\ell=1}^L
\langle\lambda,u_sb_\ell\rangle R_sB_\ell
\notag\\
&\quad+
\sum_{s=1}^T\sum_{t=1}^T
\langle\lambda,u_sv_t\rangle R_sS_t.
\label{eq:lambda_combination_general}
\end{align}
Thus, a linear combination cancels all nuisance terms precisely when $\langle\lambda,n\rangle=0$ for every $n\in\Ncal$.

\begin{proof}[Proof of Theorem~\ref{thm:decodability_condition}]
Suppose first that the desired entries are linearly independent modulo $\Ncal$. Then
\[
\Ncal
\oplus
\Span\{a_{k'}b_{\ell'}:k'\in[K],\ \ell'\in[L]\}
\]
is a direct sum. Fix $(k,\ell)\in[K]\times[L]$ and define a linear functional on this space that vanishes on $\Ncal$ and returns the coefficient of $a_kb_\ell$ in the desired component.

Extend this functional to all of $\Fq^{[N]}$. Every linear functional on $\Fq^{[N]}$ has the form $f\mapsto\langle\lambda,f\rangle$ for some $\lambda:[N]\to\Fq$. Hence there is a vector $\lambda^{(k,\ell)}$ such that
\[
\langle\lambda^{(k,\ell)},a_{k'}b_{\ell'}\rangle
=
\begin{cases}
1, & (k',\ell')=(k,\ell),\\
0, & (k',\ell')\neq(k,\ell),
\end{cases}
\]
and $\langle\lambda^{(k,\ell)},n\rangle=0$ for every $n\in\Ncal$. Substituting this vector into \eqref{eq:lambda_combination_general} gives
\[
A_kB_\ell
=
\sum_{i=1}^N\lambda^{(k,\ell)}(i)Z_i.
\]
This holds for every $(k,\ell)$, so the scheme is linearly decodable.

Conversely, suppose that the scheme is linearly decodable. For each $(k,\ell)$, let $\lambda^{(k,\ell)}$ satisfy
\[
A_kB_\ell
=
\sum_{i=1}^N\lambda^{(k,\ell)}(i)Z_i
\]
for every choice of the data blocks and masks.

Set all masks to zero and all data blocks to zero except $A_{k'}$ and $B_{\ell'}$. Choosing these two blocks so that their product is nonzero gives
\[
\langle\lambda^{(k,\ell)},a_{k'}b_{\ell'}\rangle
=
\begin{cases}
1, & (k',\ell')=(k,\ell),\\
0, & (k',\ell')\neq(k,\ell).
\end{cases}
\]

Next, set all variables to zero except one pair $A_{k'}$ and $S_t$. The desired product $A_kB_\ell$ is then zero, so choosing the two nonzero matrices with a nonzero product gives $\langle\lambda^{(k,\ell)},a_{k'}v_t\rangle=0$. Applying the same argument to one pair $R_s,B_{\ell'}$ and then to one pair $R_s,S_t$ gives
\[
\langle\lambda^{(k,\ell)},u_sb_{\ell'}\rangle=0
\qquad\text{and}\qquad
\langle\lambda^{(k,\ell)},u_sv_t\rangle=0.
\]
Since the nuisance entries span $\Ncal$, the vector $\lambda^{(k,\ell)}$ annihilates every function in $\Ncal$.

Now suppose that
\[
\sum_{k=1}^K\sum_{\ell=1}^Lc_{k\ell}a_kb_\ell
\in\Ncal.
\]
Taking the inner product with $\lambda^{(k,\ell)}$ gives $c_{k\ell}=0$. Since this holds for every $(k,\ell)$, all coefficients are zero. The desired entries are therefore linearly independent modulo $\Ncal$.
\end{proof}

\begin{corollary}
\label{cor:dimension_bound}
Every linearly decodable function table satisfies
\[
N\geq KL+\dim\Ncal.
\]
\end{corollary}

\begin{proof}
Theorem~\ref{thm:decodability_condition} gives $KL$ linearly independent classes in $\Fq^{[N]}/\Ncal$. Since this quotient has dimension $N-\dim\Ncal$, we have $KL\leq N-\dim\Ncal$, which is equivalent to the stated bound.
\end{proof}

\subsection{Privacy}

Fix a set of workers $\tau$. Privacy holds exactly when, on each side, the data coefficient columns restricted to $\tau$ lie in the span of the corresponding mask coefficient columns.

\begin{proof}[Proof of Theorem~\ref{thm:privacy_condition}]
Suppose first that
\[
\rank[\mathsf A_\tau\ \mathsf U_\tau]
=
\rank\mathsf U_\tau
\qquad\text{and}\qquad
\rank[\mathsf B_\tau\ \mathsf V_\tau]
=
\rank\mathsf V_\tau.
\]
The first equality means that every column of $\mathsf A_\tau$ lies in the column span of $\mathsf U_\tau$. Thus, for every $k\in[K]$, there are scalars $c_{1k},\ldots,c_{Tk}$ such that
\[
a_k(i)=\sum_{t=1}^Tc_{tk}u_t(i)
\qquad\text{for every }i\in\tau.
\]
Consequently, for every $i\in\tau$,
\begin{align*}
X_i
&=
\sum_{k=1}^Ka_k(i)A_k
+
\sum_{t=1}^Tu_t(i)R_t\\
&=
\sum_{t=1}^Tu_t(i)
\left(
R_t+\sum_{k=1}^Kc_{tk}A_k
\right).
\end{align*}

Fix any values of $A$ and $B$. Since $R_1,\ldots,R_T$ are independent and uniform, the shifted matrices $R_t+\sum_{k=1}^Kc_{tk}A_k$, for $t\in[T]$, have the same joint distribution as $R_1,\ldots,R_T$. It follows that the conditional distribution of $X_\tau$ does not depend on $A$.

The second rank equality gives scalars $d_{t\ell}$ such that
\[
b_\ell(i)=\sum_{t=1}^Td_{t\ell}v_t(i)
\qquad\text{for every }i\in\tau.
\]
The same argument shows that the conditional distribution of $Y_\tau$ does not depend on $B$. Since the two mask families are independent, the joint conditional distribution of $(X_\tau,Y_\tau)$ is the same for every value of $(A,B)$. Therefore $I(A,B;X_\tau,Y_\tau)=0$ for every joint distribution of $(A,B)$.

Conversely, suppose that $I(A,B;X_\tau,Y_\tau)=0$ for every joint distribution of $(A,B)$. We show that the first rank equality must hold.

Suppose otherwise. Then some column of $\mathsf A_\tau$ does not belong to the column span of $\mathsf U_\tau$. Hence there is a vector $\lambda\in\Fq^\tau$ such that
\[
\lambda^{\mathsf T}\mathsf U_\tau=0
\qquad\text{but}\qquad
\lambda^{\mathsf T}\mathsf A_\tau\neq0.
\]
Choose $k\in[K]$ such that
\[
c=\sum_{i\in\tau}\lambda(i)a_k(i)\neq0.
\]

Fix $B$, let all blocks of $A$ other than $A_k$ be zero, and let $A_k$ be equally likely to be either the zero matrix or a fixed nonzero matrix $E$. The masks cancel in the linear combination
\[
\sum_{i\in\tau}\lambda(i)X_i=cA_k.
\]
This value is zero when $A_k=0$ and equals the nonzero matrix $cE$ when $A_k=E$. The observations $X_\tau$ therefore reveal which value of $A_k$ was chosen, contradicting privacy. We conclude that every column of $\mathsf A_\tau$ belongs to the column span of $\mathsf U_\tau$, and hence $\rank[\mathsf A_\tau\ \mathsf U_\tau]
=
\rank\mathsf U_\tau$. The same argument on the $B$-side gives
$\rank[\mathsf B_\tau\ \mathsf V_\tau]
=
\rank\mathsf V_\tau$.
\end{proof}

\begin{proof}[Proof of Corollary~\ref{cor:mds_privacy}]
Let $\tau\subseteq[N]$ with $|\tau|\leq T$. The $T$-MDS mask condition gives $\rank\mathsf U_\tau=|\tau|$, so the columns of $\mathsf U_\tau$ span all of $\Fq^\tau$. Every column of $\mathsf A_\tau$ therefore belongs to this span, and $\rank[\mathsf A_\tau\ \mathsf U_\tau]
=
\rank\mathsf U_\tau$.
The same argument gives $\rank[\mathsf B_\tau\ \mathsf V_\tau]
=
\rank\mathsf V_\tau$.
Theorem~\ref{thm:privacy_condition} now shows that the observations of every set of at most $T$ workers are independent of $(A,B)$. Hence the scheme is $T$-private.
\end{proof}

Thus, the MDS mask condition is sufficient but not necessary: dependent mask rows are allowed whenever every linear relation among them also annihilates the data coefficients.

\section{Optimality for $T=1$}
\label{sec:t1_constructions}

We now prove Theorem~\ref{thm:t1_optimality}. For achievability, we choose coordinate functions on a sufficiently large set and then select exactly the number of points needed for the workers. For the converse, we normalize the two mask functions to $\one$ and show that decodability forces the nuisance space to have dimension at least $K+L$. Over $\F_2$, one additional nuisance dimension is unavoidable.

\subsection{Achievability}

We begin with two elementary facts about functions over finite fields.

\begin{lemma}[Reduced monomials]
\label{lem:reduced_monomials}
Let $\Omega=\Fq^d$, with coordinate functions $x_1,\ldots,x_d$. The monomials $x_1^{e_1}\cdots x_d^{e_d}$ with $0\leq e_j\leq q-1$ form a basis of the functions from $\Omega$ to $\Fq$.
\end{lemma}
\begin{proof}
If a polynomial $p$ of degree at most $q-1$ in each variable vanishes at every point of $\Fq^d$, then $p=0$ by \cite[Lemma 1]{ALON_1999}. Thus, the $q^d$ listed monomials are linearly independent. As the space of functions from $\Omega$ to $\Fq$ has dimension $|\Omega|=q^d$, the monomials form a basis.
\end{proof}

\begin{lemma}[Selecting worker points]
\label{lem:selecting_worker_points}
Let $f_1,\ldots,f_m$ be linearly independent functions on a finite set $\Omega$. Then there are points $\omega_1,\ldots,\omega_m\in\Omega$ such that the restrictions of these functions to the selected points remain linearly independent.
\end{lemma}

\begin{proof}
Form the matrix whose rows are indexed by the points of $\Omega$, whose columns are indexed by $f_1,\ldots,f_m$, and whose $(\omega,j)$ entry is $f_j(\omega)$. Since the functions are linearly independent, the matrix has column rank $m$. It therefore contains $m$ linearly independent rows. Choosing the corresponding points proves the result.
\end{proof}

We first consider nonbinary fields. The construction uses one nonconstant function on both sides of the table. This reduces the nuisance dimension by one, while its square remains a distinct function.

\begin{proof}[Achievability for $q\geq3$ in Theorem~\ref{thm:t1_optimality}]
Let $\Omega=\Fq^{K+L-1}$, with coordinate functions $z,x_1,\ldots,x_{K-1},y_1,\ldots,y_{L-1}$. If $K=1$ or $L=1$, the corresponding list is empty.

On $\Omega$, set $u_1=v_1=\one$, $a_1=b_1=z$, $a_k=x_{k-1}$ for $2\leq k\leq K$, and $b_\ell=y_{\ell-1}$ for $2\leq\ell\leq L$.

The nuisance functions are $\one,z,x_1,\ldots,x_{K-1},y_1,\ldots,y_{L-1}$. There are $K+L$ of them. The desired functions are $z^2$, the functions $x_iz$ for $i\in[K-1]$, the functions $zy_j$ for $j\in[L-1]$, and the functions $x_iy_j$ for $i\in[K-1]$ and $j\in[L-1]$. There are $1+(K-1)+(L-1)+(K-1)(L-1)=KL$ desired functions.

Since $q\geq3$, all $KL+K+L$ functions in these two lists are distinct reduced monomials. Lemma~\ref{lem:reduced_monomials} shows that they are linearly independent on $\Omega$. Lemma~\ref{lem:selecting_worker_points} therefore gives $N=KL+K+L$ points on which they remain linearly independent. Assign one selected point to each worker and evaluate the chosen row and column functions there.

The desired entries are therefore linearly independent modulo $\Ncal$, so the scheme is linearly decodable.

The mask coefficient is one at every worker on both sides. Hence the two mask functions satisfy the $1$-MDS mask condition, and the scheme is private against any single worker.
\end{proof}

Over $\F_2$, this construction is impossible because every function $z:[N]\to\F_2$ satisfies $z^2=z$. We instead use separate coordinate functions on the two sides.

\begin{proof}[Achievability over $\F_2$ in Theorem~\ref{thm:t1_optimality}]
Let $\Omega=\F_2^{K+L}$, with coordinate functions $x_1,\ldots,x_K,y_1,\ldots,y_L$.

On $\Omega$, set $u_1=v_1=\one$, $a_k=x_k$ for every $k\in[K]$, and $b_\ell=y_\ell$ for every $\ell\in[L]$.

The nuisance functions are $\one,x_1,\ldots,x_K,y_1,\ldots,y_L$, and the desired functions are $x_ky_\ell$ for $k\in[K]$ and $\ell\in[L]$.

All $KL+K+L+1$ functions are distinct squarefree monomials. Lemma~\ref{lem:reduced_monomials} shows that they are linearly independent on $\Omega$. Lemma~\ref{lem:selecting_worker_points} therefore gives $N=KL+K+L+1$ points on which they remain linearly independent. Assign one selected point to each worker and evaluate the chosen row and column functions there.

The desired entries are therefore linearly independent modulo $\Ncal$, so the scheme is linearly decodable.

Both masks have coefficient one at every worker, so the scheme is private against any single worker.
\end{proof}

\subsection{Converse}

We next show that the preceding constructions use the minimum possible number of workers. The first step is to normalize the mask functions.

\begin{lemma}[Normalizing the masks for $T=1$]
\label{lem:t1_normalization}
Every private and linearly decodable function table with $T=1$ can be replaced, without increasing the number of workers, by one with $u_1=v_1=\one$.
\end{lemma}

\begin{proof}
Fix a worker $i$. If $u_1(i)=0$, then the privacy condition for $\tau=\{i\}$ gives $\rank[\mathsf A_\tau\ \mathsf U_\tau]=\rank\mathsf U_\tau=0$. Thus $a_k(i)=0$ for every $k\in[K]$. It follows that $X_i=0$ and hence $Z_i=0$, so the worker can be removed without affecting privacy or decodability. Similarly, if $v_1(i)=0$, privacy gives $b_\ell(i)=0$ for every $\ell\in[L]$, so $Y_i=0$ and the worker can again be removed.

We may therefore assume that $u_1(i)$ and $v_1(i)$ are nonzero at every worker. Define
\[
\widetilde X_i
=
\frac{X_i}{u_1(i)}
=
R_1+\sum_{k=1}^K\frac{a_k(i)}{u_1(i)}A_k
\]
and
\[
\widetilde Y_i
=
\frac{Y_i}{v_1(i)}
=
S_1+\sum_{\ell=1}^L\frac{b_\ell(i)}{v_1(i)}B_\ell.
\]
The corresponding response is $\widetilde Z_i=Z_i/(u_1(i)v_1(i))$. These are fixed invertible rescalings of the workers' observations and responses, so they preserve privacy and linear decodability. The resulting mask functions are $u_1=v_1=\one$.
\end{proof}

Once the masks are constant, the nuisance space is generated by $\one$ together with the data functions from the two sides. The next lemma bounds how much the two sides can overlap.

\begin{lemma}[The constant-mask case]
\label{lem:constant_mask_lower_bound}
Let $u_1=v_1=\one$. If the function table is linearly decodable, then $\dim\Ncal\geq K+L$. If $q=2$, then $\dim\Ncal\geq K+L+1$.
\end{lemma}

\begin{proof}
Here $\Ncal=\Span\{\one,a_1,\ldots,a_K,b_1,\ldots,b_L\}$. Let $\mathcal P=\Span\{\one,a_1,\ldots,a_K\}$ and $\mathcal Q=\Span\{\one,b_1,\ldots,b_L\}$, so $\Ncal=\mathcal P+\mathcal Q$.

We first show that $\one,a_1,\ldots,a_K$ are linearly independent. Suppose that $c_0\one+\sum_{k=1}^Kc_ka_k=0$. If some $c_k$ is nonzero, then, for any fixed $\ell\in[L]$,
\[
\sum_{k=1}^Kc_ka_kb_\ell=-c_0b_\ell\in\Ncal.
\]
This is a nonzero linear combination of desired entries that belongs to $\Ncal$, contradicting decodability. Hence every $c_k$ is zero, and then $c_0=0$. Therefore $\dim\mathcal P=K+1$. The same argument gives $\dim\mathcal Q=L+1$.

We next show that $\dim(\mathcal P\cap\mathcal Q)\leq2$. Suppose otherwise. Since $\one\in\mathcal P\cap\mathcal Q$, there are functions $h,g\in\mathcal P\cap\mathcal Q$ such that $\one,h,g$ are linearly independent. Write
\begin{align*}
h
&=
\alpha\one+\sum_{k=1}^Kc_ka_k
=
\beta\one+\sum_{\ell=1}^Ld_\ell b_\ell,\\
g
&=
\gamma\one+\sum_{k=1}^Ke_ka_k
=
\delta\one+\sum_{\ell=1}^Lf_\ell b_\ell.
\end{align*}
The vectors $c=(c_k)_k$ and $e=(e_k)_k$ are linearly independent, since otherwise $\one,h,g$ would be linearly dependent. Similarly, $d=(d_\ell)_\ell$ and $f=(f_\ell)_\ell$ are linearly independent.

Now
\begin{align*}
&\left(\sum_{k=1}^Kc_ka_k\right)
 \left(\sum_{\ell=1}^Lf_\ell b_\ell\right)
-
\left(\sum_{k=1}^Ke_ka_k\right)
 \left(\sum_{\ell=1}^Ld_\ell b_\ell\right)\\
&\qquad
=
(h-\alpha\one)(g-\delta\one)
-
(g-\gamma\one)(h-\beta\one)\\
&\qquad
=
(\gamma-\delta)h
+
(\beta-\alpha)g
+
(\alpha\delta-\gamma\beta)\one
\in\Ncal.
\end{align*}
The coefficient of $a_kb_\ell$ on the left-hand side is $c_kf_\ell-e_kd_\ell$. These coefficients cannot all be zero. Otherwise, $f_\ell c=d_\ell e$ for every $\ell$. Since $c$ and $e$ are linearly independent, this would force $f_\ell=d_\ell=0$ for every $\ell$, contradicting the linear independence of $d$ and $f$. Thus the left-hand side is a nonzero linear combination of desired entries that belongs to $\Ncal$, again contradicting decodability. Therefore $\dim(\mathcal P\cap\mathcal Q)\leq2$.

It follows that
\[
\dim\Ncal
=
\dim\mathcal P+\dim\mathcal Q-\dim(\mathcal P\cap\mathcal Q)
\geq K+L.
\]

Suppose now that $q=2$. We show that $\mathcal P\cap\mathcal Q=\Span\{\one\}$. Otherwise, choose a nonconstant function $h\in\mathcal P\cap\mathcal Q$ and write
\[
h
=
\alpha\one+\sum_{k=1}^Kc_ka_k
=
\beta\one+\sum_{\ell=1}^Ld_\ell b_\ell.
\]
Since $h$ is nonconstant, both $(c_k)_k$ and $(d_\ell)_\ell$ are nonzero. Hence there are indices $k$ and $\ell$ for which $c_kd_\ell\neq0$, so
\[
\left(\sum_{k=1}^Kc_ka_k\right)
\left(\sum_{\ell=1}^Ld_\ell b_\ell\right)
\]
is a nonzero linear combination of desired entries. On the other hand,
\begin{align*}
\left(\sum_{k=1}^Kc_ka_k\right)
\left(\sum_{\ell=1}^Ld_\ell b_\ell\right)
&=
(h-\alpha\one)(h-\beta\one)\\
&=
h^2-(\alpha+\beta)h+\alpha\beta\one.
\end{align*}
Every function from $[N]$ to $\F_2$ satisfies $h^2=h$, so the right-hand side belongs to $\Span\{\one,h\}\subseteq\Ncal$. This contradicts decodability. Thus $\mathcal P\cap\mathcal Q=\Span\{\one\}$, and $\dim\Ncal=K+L+1$.
\end{proof}

\begin{proof}[Converse in Theorem~\ref{thm:t1_optimality}]
By Lemma~\ref{lem:t1_normalization}, we may assume that $u_1=v_1=\one$ without increasing the number of workers. Lemma~\ref{lem:constant_mask_lower_bound} gives $\dim\Ncal\geq K+L$. Corollary~\ref{cor:dimension_bound} then gives $N\geq KL+K+L$.

When $q=2$, Lemma~\ref{lem:constant_mask_lower_bound} gives the stronger inequality $\dim\Ncal\geq K+L+1$, and hence $N\geq KL+K+L+1$. These bounds match the two constructions and complete the proof.
\end{proof}

\section{General Lower Bound}
\label{sec:general_lower_bounds}

We prove Theorem~\ref{thm:general_lower_bound} by reducing an arbitrary private function table to the constant-mask case. Over a sufficiently large extension field, each mask space contains a function that is nonzero at every worker. Dividing by these functions produces a decodable table with one constant mask on each side.

\begin{lemma}[Choosing a nonvanishing function]
\label{lem:nonvanishing_function}
Let $\mathbb E$ be a finite field, and let $\mathcal W\subseteq\mathbb E^{[N]}$ be a vector space. Suppose that, for every $i\in[N]$, some function in $\mathcal W$ is nonzero at $i$. If $|\mathbb E|>N$, then $\mathcal W$ contains a function that is nonzero at every worker.
\end{lemma}

\begin{proof}
Let $d=\dim_{\mathbb E}\mathcal W$. For each $i\in[N]$, let $H_i=\{w\in\mathcal W:w(i)=0\}$. By assumption, evaluation at $i$ is not identically zero on $\mathcal W$, so $H_i$ is a proper hyperplane and contains $|\mathbb E|^{d-1}$ functions. Therefore
\[
\left|\bigcup_{i=1}^NH_i\right|
\leq
N|\mathbb E|^{d-1}
<
|\mathbb E|^d
=
|\mathcal W|.
\]
Thus, some $w\in\mathcal W$ belongs to none of the sets $H_i$, and hence $w(i)\neq0$ for every $i\in[N]$.
\end{proof}

\begin{proof}[Proof of Theorem~\ref{thm:general_lower_bound}]
We first remove workers whose responses always vanish. Suppose that $u_t(i)=0$ for every $t\in[T]$. Applying Theorem~\ref{thm:privacy_condition} to $\tau=\{i\}$ gives $\rank[\mathsf A_\tau\ \mathsf U_\tau]=\rank\mathsf U_\tau=0$, so $a_k(i)=0$ for every $k\in[K]$. Thus $X_i=Z_i=0$, and worker $i$ can be removed without affecting privacy or decodability. The same argument applies if $v_t(i)=0$ for every $t\in[T]$.

After removing all such workers, relabel the remaining workers and continue to write $N$ for their number. Proving the bound for the smaller scheme also proves it for the original scheme. For every remaining worker $i$, some $u_t(i)$ and some $v_t(i)$ are nonzero.

Choose a finite extension field $\mathbb E$ of $\Fq$ with $|\mathbb E|>N$, and regard all coefficient functions as functions from $[N]$ to $\mathbb E$. Let $\Acal_{\mathbb E}$, $\Ucal_{\mathbb E}$, $\Bcal_{\mathbb E}$, and $\Vcal_{\mathbb E}$ denote their spans over $\mathbb E$, and let $\Ncal_{\mathbb E}$ be the $\mathbb E$-span of the original nuisance entries. Then
\[
\Ncal_{\mathbb E}
=
\Acal_{\mathbb E}\Vcal_{\mathbb E}
+
\Ucal_{\mathbb E}\Bcal_{\mathbb E}
+
\Ucal_{\mathbb E}\Vcal_{\mathbb E}.
\]

Let $\mathsf G$ be a matrix whose columns span $\Ncal$, and let $\mathsf D$ be the matrix whose columns are the $KL$ desired entries. Decodability gives $\rank[\mathsf G\ \mathsf D]=\rank\mathsf G+KL$. Extending the scalars from $\Fq$ to $\mathbb E$ does not change these ranks, so the desired entries remain linearly independent modulo $\Ncal_{\mathbb E}$.

The spaces $\Ucal_{\mathbb E}$ and $\Vcal_{\mathbb E}$ satisfy the hypothesis of Lemma~\ref{lem:nonvanishing_function}. Hence there are functions $u\in\Ucal_{\mathbb E}$ and $v\in\Vcal_{\mathbb E}$ such that $u(i)\neq0$ and $v(i)\neq0$ for every $i\in[N]$.

Define $\widetilde a_k=a_k/u$ for $k\in[K]$ and $\widetilde b_\ell=b_\ell/v$ for $\ell\in[L]$, where division is pointwise. Consider the auxiliary function table over $\mathbb E$ with data functions $\widetilde a_1,\ldots,\widetilde a_K$ and $\widetilde b_1,\ldots,\widetilde b_L$, and with $\one$ as the single mask function on both sides. Its nuisance space is
\[
\widetilde{\Ncal}
=
\Span_{\mathbb E}
\{\one,\widetilde a_1,\ldots,\widetilde a_K,
\widetilde b_1,\ldots,\widetilde b_L\}.
\]

Multiplication by the nonvanishing function $uv$ is an invertible linear map on $\mathbb E^{[N]}$. Moreover,
\begin{align*}
(uv)\widetilde{\Ncal}
&=
\Span_{\mathbb E}
\{uv,a_1v,\ldots,a_Kv,ub_1,\ldots,ub_L\}\\
&\subseteq
\Ucal_{\mathbb E}\Vcal_{\mathbb E}
+
\Acal_{\mathbb E}\Vcal_{\mathbb E}
+
\Ucal_{\mathbb E}\Bcal_{\mathbb E}\\
&=
\Ncal_{\mathbb E}.
\end{align*}

Suppose that a linear combination of the desired entries of the auxiliary table belongs to $\widetilde{\Ncal}$. Multiplying by $uv$ gives
\[
\sum_{k=1}^K\sum_{\ell=1}^L
c_{k\ell}a_kb_\ell
\in
\Ncal_{\mathbb E}.
\]
Since the original desired entries remain linearly independent modulo $\Ncal_{\mathbb E}$, we have $c_{k\ell}=0$ for every $k$ and $\ell$. The auxiliary table is therefore linearly decodable.

Applying Lemma~\ref{lem:constant_mask_lower_bound} over $\mathbb E$ gives $\dim_{\mathbb E}\widetilde{\Ncal}\geq K+L$. Corollary~\ref{cor:dimension_bound} then gives $N\geq KL+\dim_{\mathbb E}\widetilde{\Ncal}\geq KL+K+L$, proving the theorem.
\end{proof}

\section{Field Feasibility}
\label{sec:field_feasibility_proofs}

This section determines exactly which finite fields can support a private and linearly decodable scheme. We first show that every such scheme contains sets of workers on which the mask coefficients are MDS. We then use the required MDS codes to construct a scheme.

\subsection{Necessary MDS Restrictions}

Decodability first gives a separation between the data and mask spaces. On each side, no nonzero linear combination of the data functions can belong to the mask space.

\begin{lemma}[Data--mask separation]
\label{lem:data_mask_separation}
In every linearly decodable function table, the quotient classes $a_1+\Ucal,\ldots,a_K+\Ucal$ are linearly independent in $\Fq^{[N]}/\Ucal$. Similarly, the quotient classes $b_1+\Vcal,\ldots,b_L+\Vcal$ are linearly independent in $\Fq^{[N]}/\Vcal$.
\end{lemma}

\begin{proof}
Suppose that $\sum_{k=1}^Kc_ka_k\in\Ucal$. Fix any $\ell\in[L]$. Multiplying by $b_\ell$ gives $\sum_{k=1}^Kc_ka_kb_\ell\in\Ucal\Bcal\subseteq\Ncal$. The left-hand side is a linear combination of desired entries. By Theorem~\ref{thm:decodability_condition}, the desired entries are linearly independent modulo $\Ncal$, so $c_k=0$ for every $k\in[K]$.

The proof on the $B$-side is the same. If $\sum_{\ell=1}^Ld_\ell b_\ell\in\Vcal$, multiply by any $a_k$ and apply Theorem~\ref{thm:decodability_condition}.
\end{proof}

In particular, Lemma~\ref{lem:data_mask_separation} gives $\dim\Acal=K$, $\Acal\cap\Ucal=\{0\}$, $\dim\Bcal=L$, and $\Bcal\cap\Vcal=\{0\}$.

\begin{proof}[Proof of Theorem~\ref{thm:necessary_mds_restrictions}]
Let $\mathsf A=\mathsf A_{[N]}$ and $\mathsf U=\mathsf U_{[N]}$, and set $r=\rank\mathsf U$. By Lemma~\ref{lem:data_mask_separation}, the $K$ columns of $\mathsf A$ are linearly independent modulo the column span of $\mathsf U$. Therefore $\rank[\mathsf A\ \mathsf U]=K+r$.

We first prove that $r=T$. Suppose instead that $r<T$. Since $K\geq1$, the matrix $[\mathsf A\ \mathsf U]$ has rank at least $r+1$. We may therefore choose a set $\tau\subseteq[N]$ of $r+1$ rows such that $\rank[\mathsf A_\tau\ \mathsf U_\tau]=r+1$. Since $|\tau|=r+1\leq T$, privacy and Theorem~\ref{thm:privacy_condition} give $\rank[\mathsf A_\tau\ \mathsf U_\tau]=\rank\mathsf U_\tau\leq r$, a contradiction. Hence $\rank\mathsf U=T$.

It follows that $\rank[\mathsf A\ \mathsf U]=K+T$. Choose a set $\tau_A\subseteq[N]$ of $K+T$ linearly independent rows of this matrix. If $\tau\subseteq\tau_A$ and $|\tau|\leq T$, then the rows of $[\mathsf A_\tau\ \mathsf U_\tau]$ are linearly independent, so $\rank[\mathsf A_\tau\ \mathsf U_\tau]=|\tau|$. Privacy gives $\rank\mathsf U_\tau=|\tau|$.

Thus every set of at most $T$ rows of $\mathsf U_{\tau_A}$ is linearly independent. Equivalently, $\mathsf U_{\tau_A}^{\mathsf T}$ generates an $[K+T,T]$ linear MDS code.

The same argument on the $B$-side gives $\rank\mathsf V_{[N]}=T$ and a set $\tau_B\subseteq[N]$ of size $L+T$ such that $\mathsf V_{\tau_B}^{\mathsf T}$ generates an $[L+T,T]$ linear MDS code.
\end{proof}

The two MDS restrictions may involve different sets of workers, but their sizes still give an additional worker lower bound.

\begin{proof}[Proof of Corollary~\ref{cor:combined_worker_lower_bound}]
Theorem~\ref{thm:necessary_mds_restrictions} gives $N\geq K+T$ and $N\geq L+T$, and hence $N\geq T+\max\{K,L\}$. Combining this with Theorem~\ref{thm:general_lower_bound} proves the result.
\end{proof}

\subsection{The Cartesian Construction}

The MDS restrictions above are also sufficient. As in Section~\ref{sec:cartesian_example}, the Cartesian construction pairs each of $K+T$ encodings on the $A$-side with each of $L+T$ encodings on the $B$-side.

\begin{proof}[Proof of Theorem~\ref{thm:field_feasibility}]
Suppose first that a private and linearly decodable scheme exists. Theorem~\ref{thm:necessary_mds_restrictions} gives an $[K+T,T]$ linear MDS code and an $[L+T,T]$ linear MDS code. One of these is an $[\max\{K,L\}+T,T]$ linear MDS code, proving necessity.

Conversely, suppose that an $[\max\{K,L\}+T,T]$ linear MDS code exists. Deleting coordinates gives an $[K+T,T]$ linear MDS code and an $[L+T,T]$ linear MDS code. Let $\mathsf U_0\in\Fq^{(K+T)\times T}$ and $\mathsf V_0\in\Fq^{(L+T)\times T}$ be the transposes of generator matrices of these two codes. Thus every set of at most $T$ rows of either matrix is linearly independent.

Complete the columns of $\mathsf U_0$ to a basis of $\Fq^{K+T}$. That is, choose $\mathsf A_0\in\Fq^{(K+T)\times K}$ so that $P=[\mathsf A_0\ \mathsf U_0]$ is invertible. Similarly, choose $\mathsf B_0\in\Fq^{(L+T)\times L}$ so that $Q=[\mathsf B_0\ \mathsf V_0]$ is invertible.

Index the workers by $[K+T]\times[L+T]$. At worker $(i,j)$, set $a_k(i,j)=\mathsf A_0(i,k)$ and $u_t(i,j)=\mathsf U_0(i,t)$ on the $A$-side, and set $b_\ell(i,j)=\mathsf B_0(j,\ell)$ and $v_t(i,j)=\mathsf V_0(j,t)$ on the $B$-side. Thus the $A$-side coefficients depend only on $i$, while the $B$-side coefficients depend only on $j$. The construction has $(K+T)(L+T)$ workers.

We first verify decodability. The columns of $P$ form a basis of the functions on $[K+T]$, and the columns of $Q$ form a basis of the functions on $[L+T]$. For a column $f$ of $P$ and a column $g$ of $Q$, the corresponding entry of the function table is the function $f\otimes g$ defined by $(f\otimes g)(i,j)=f(i)g(j)$. As in the Cartesian example, the $(K+T)(L+T)$ functions $f\otimes g$ form a basis of the functions on $[K+T]\times[L+T]$.

Hence all entries of the Cartesian function table are linearly independent. In particular, the desired entries are linearly independent modulo the nuisance space, so Theorem~\ref{thm:decodability_condition} gives linear decodability. The nuisance space has dimension $(K+T)(L+T)-KL=T(K+L)+T^2$.

We next verify privacy. Let $\tau\subseteq[K+T]\times[L+T]$ be any set of at most $T$ workers, and let $I=\{i:(i,j)\in\tau\text{ for some }j\}$ be its set of distinct first coordinates. The $A$-side coefficients depend only on $i$, so workers with the same first coordinate give identical rows of $[\mathsf A_\tau\ \mathsf U_\tau]$. Its distinct rows are exactly the rows of $P$ indexed by $I$. Since $P$ is invertible, these rows are linearly independent, and therefore $\rank[\mathsf A_\tau\ \mathsf U_\tau]=|I|$.

Since $|I|\leq|\tau|\leq T$, the MDS property of $\mathsf U_0$ implies that its rows indexed by $I$ are also linearly independent. Repetitions do not change their span, so $\rank\mathsf U_\tau=|I|$. Thus $\rank[\mathsf A_\tau\ \mathsf U_\tau]=\rank\mathsf U_\tau$.

Similarly, let $J=\{j:(i,j)\in\tau\text{ for some }i\}$ be the set of distinct second coordinates. The $B$-side coefficients depend only on $j$. Applying the same argument to $Q$ and $\mathsf V_0$ gives $\rank[\mathsf B_\tau\ \mathsf V_\tau]=\rank\mathsf V_\tau=|J|$.

Theorem~\ref{thm:privacy_condition} therefore gives $T$-privacy. Repeated rows cause no problem because whenever a mask row repeats, the corresponding data row repeats with it.
\end{proof}

\subsection{Field-Size Consequences}

The field-size results follow from a standard bound on the length of a linear MDS code.

\begin{lemma}[Elementary MDS length bound]
\label{lem:elementary_mds_length}
If an $[n,k]$ linear MDS code over $\Fq$ exists with $2\leq k<n$, then $n-k\leq q-1$.
\end{lemma}

\begin{proof}
After permuting columns and applying row operations, write a generator matrix in systematic form as $G=[I_k\ P]$. Every entry of $P$ is nonzero. Indeed, if an entry in row $i$ were zero, the corresponding column of $P$ together with the $k-1$ columns of $I_k$ other than the $i$th column would be linearly dependent.

For each column of $P$, consider the ratio of its first entry to its second entry. These ratios must be distinct. Otherwise, two columns of $P$ would have proportional first two entries. A nonzero linear combination of these two columns would then vanish in the first two coordinates. The remaining coordinates could be canceled using the identity columns $e_3,\ldots,e_k$, giving a dependence among $k$ columns of $G$ and contradicting the MDS property.

There are only $q-1$ possible nonzero ratios. Hence $P$ has at most $q-1$ columns, and therefore $n-k\leq q-1$.
\end{proof}

\begin{proof}[Proof of Corollary~\ref{cor:field_size_consequences}]
If a scheme exists, Theorem~\ref{thm:field_feasibility} gives an $[\max\{K,L\}+T,T]$ linear MDS code. When $T\geq2$, Lemma~\ref{lem:elementary_mds_length} gives $\max\{K,L\}\leq q-1$, and hence $q\geq\max\{K,L\}+1$.

If also $\max\{K,L\}\geq2$, the dual code is an $[\max\{K,L\}+T,\max\{K,L\}]$ linear MDS code. Applying Lemma~\ref{lem:elementary_mds_length} to the dual gives $T\leq q-1$, and hence $q\geq T+1$.

Conversely, suppose that $q\geq\max\{K,L\}+T-1$. Then $\max\{K,L\}+T\leq q+1$, so an extended Reed--Solomon code gives an $[\max\{K,L\}+T,T]$ linear MDS code over $\Fq$. Theorem~\ref{thm:field_feasibility} therefore gives a private and linearly decodable scheme.
\end{proof}

\begin{proof}[Proof of Corollary~\ref{cor:t2_field_feasibility}]
For $T=2$, the necessary condition $q\geq\max\{K,L\}+1$ and the sufficient condition $q\geq\max\{K,L\}+T-1$ in Corollary~\ref{cor:field_size_consequences} coincide.
\end{proof}

We finish by proving the stronger field-size restriction that arises when the mask coefficients are MDS on the complete worker set.

\begin{proof}[Proof of Proposition~\ref{prop:mds_field_size_bound}]
Suppose first that $u_1,\ldots,u_T$ satisfy the $T$-MDS mask condition, and let $\mathbf u_i=(u_1(i),\ldots,u_T(i))\in\Fq^T$ be the mask-coefficient row of worker $i$.

If $N\leq T-2$, then $N\leq q+T-1$ is immediate. We may therefore assume that $N\geq T-1$ and fix any $T-2$ rows. Let $\mathcal W$ be their span. When $T=2$, no rows are fixed and we take $\mathcal W=\{0\}$. The fixed rows are linearly independent, so $\dim\mathcal W=T-2$, and the quotient space $\Fq^T/\mathcal W$ has dimension two.

The image of every remaining row in this quotient is nonzero. Moreover, the images of any two remaining rows lie in different one-dimensional subspaces. Otherwise, those two rows together with the fixed $T-2$ rows would be linearly dependent, contradicting the $T$-MDS mask condition.

A two-dimensional vector space over $\Fq$ has exactly $q+1$ one-dimensional subspaces. There are therefore at most $q+1$ remaining rows. Hence $N-(T-2)\leq q+1$, or equivalently $N\leq q+T-1$.

The proof is identical if $v_1,\ldots,v_T$ satisfy the $T$-MDS mask condition.
\end{proof}

\section[The Projective-Line Construction for T=2]
{The Projective-Line Construction for $T=2$}
\label{sec:constructions_beyond_t1}

\begin{figure*}[t]
\centering

\def\ProjectiveTableWidth{0.65\linewidth}
\def\ProjectiveSeparationWidth{0.65\linewidth}
\def\ProjectiveEncodingWidth{0.65\textwidth}

\begin{minipage}[t]{0.44\textwidth}
\vspace{0pt}
\centering

{\small\textbf{(a) Projective addition table}\par}

\medskip

\begingroup
\setlength{\arraycolsep}{4pt}
\renewcommand{\arraystretch}{1.18}
\small

\resizebox{\ProjectiveTableWidth}{!}{$
\begin{array}{c|ccc|cc}
\boxplus
& d
& \cdots
& Ld
& \cellcolor{blue!15}0
& \cellcolor{blue!15}s\\
\hline
1
& \cellcolor{red!12}d+1
& \cellcolor{red!12}\cdots
& \cellcolor{red!12}Ld+1
& 1
& 2\\
\vdots
& \cellcolor{red!12}\vdots
& \cellcolor{red!12}\ddots
& \cellcolor{red!12}\vdots
& \vdots
& \vdots\\
K
& \cellcolor{red!12}d+K
& \cellcolor{red!12}\cdots
& \cellcolor{red!12}m
& K
& d\\
\hline
\cellcolor{green!15}0
& d
& \cdots
& Ld
& \cellcolor{yellow!18}0
& 1\\
\cellcolor{green!15}d
& 2d
& \cdots
& 1
& d
& \cellcolor{yellow!18}\infty
\end{array}
$}
\endgroup

\end{minipage}
\hfill
\begin{minipage}[t]{0.52\textwidth}
\vspace{0pt}
\centering

{\small\textbf{(b) The role of $0$ and $\infty$}\par}

\medskip

\begingroup
\setlength{\arraycolsep}{6pt}
\renewcommand{\arraystretch}{1.35}
\small

\resizebox{\ProjectiveSeparationWidth}{!}{$
\begin{array}{c|c|cc}
\text{Function}
& (t,1),\ t\in H
& 0
& \infty\\
\hline
\cellcolor{red!12}h_m
& 1
& 0
& 0\\
\cellcolor{yellow!18}h_0
& 1
& \cellcolor{yellow!18}1
& 0\\
\hline
\cellcolor{red!12}h_{d+1}
& t^{d+1}
& 0
& 0\\
\cellcolor{yellow!18}h_D
& t^{d+1}
& 0
& \cellcolor{yellow!18}1
\end{array}
$}
\endgroup

\end{minipage}

\bigskip

{\small\textbf{(c) Evaluation points and worker encodings}\par}

\medskip

\begingroup
\setlength{\arraycolsep}{7pt}
\renewcommand{\arraystretch}{1.5}
\small

\resizebox{\ProjectiveEncodingWidth}{!}{$
\begin{array}{c|c|c|c}
P
& X_P
& Y_P
& Z_P\\
\hline
(t,1),\ t\in H
& \displaystyle\sum_{k=1}^{K}t^kA_k+R_1+t^dR_2
& \displaystyle\sum_{\ell=1}^{L}t^{\ell d}B_\ell+S_1+tS_2
& X_PY_P\\[1ex]
\hline
0=(0,1)
& R_1
& S_1
& \cellcolor{yellow!18}R_1S_1\\
\infty=(1,0)
& R_2
& S_2
& \cellcolor{yellow!18}R_2S_2
\end{array}
$}
\endgroup

\vspace{0.8em}

\caption{The projective-line construction for $T=2$, with $d=K+1$, $s=d(L+1)$, $m=s-1=KL+K+L$, and $D=d+s$. Panel~(a) reduces exponent sums modulo $m$ to labels $1,\ldots,m$, except that sums $0$ and $D$ receive the separate labels $0$ and $\infty$. Red entries are desired, and green and blue labels mark mask functions. Panel~(b) shows why the two yellow mask--mask entries must remain separate. Here $h_e(x,y)=x^ey^{D-e}$, and $H\subseteq\F_q^\times$ is the subgroup of order $m$, which exists when $m\mid(q-1)$. The functions $h_0$ and $h_m$ agree on $H$ but differ at $0=(0,1)$; similarly, $h_D$ and $h_{d+1}$ agree on $H$ but differ at $\infty=(1,0)$. Panel~(c) gives the worker encodings. There is one worker for each $t\in H$ and one at each additional point, giving $m+2=KL+K+L+2$ workers. The two additional workers return $R_1S_1$ and $R_2S_2$, the mask products associated with the yellow entries.}
\label{fig:projective_function_table}
\end{figure*}

We now generalize the projective-line construction in Section~\ref{sec:projective_example} to arbitrary $K$ and $L$. The construction uses the same homogeneous monomials, identities on a multiplicative subgroup, and two additional points $0$ and $\infty$.

Set $d=K+1$, $s=d(L+1)=(K+1)(L+1)$, $m=s-1=KL+K+L$, and $D=d+s$. For $0\leq e\leq d$, define $f_e(x,y)=x^ey^{d-e}$. For $0\leq e\leq s$, define $g_e(x,y)=x^ey^{s-e}$. Finally, for $0\leq e\leq D$, define $h_e(x,y)=x^ey^{D-e}$. Pointwise multiplication satisfies $f_eg_{e'}=h_{e+e'}$.

We use $a_k=f_k$ and $b_\ell=g_{\ell d}$ as the data functions, and $u_1=f_0$, $u_2=f_d$, $v_1=g_0$, and $v_2=g_s$ as the mask functions. Thus, the row labels are $1,\ldots,K,0,d$, and the column labels are $d,2d,\ldots,Ld,0,s$.

Assume that $m\mid(q-1)$, and let $H\subseteq\Fq^\times$ be the subgroup of order $m$. Index the workers by
\[
\Omega
=
\{(t,1):t\in H\}
\cup
\{(0,1),(1,0)\}
\subseteq\Fq^2.
\]
These vectors are fixed representatives of points of $\mathbb P^1(\Fq)$. We write $0=(0,1)$ and $\infty=(1,0)$. The construction therefore uses $|\Omega|=m+2=KL+K+L+2$ workers.

At a subgroup worker $(t,1)$, we have $h_e(t,1)=t^e$. Since $t^m=1$, exponents that are congruent modulo $m$ give the same values on the subgroup. Every $h_e$ with $1\leq e\leq D-1$ also vanishes at both $0$ and $\infty$. Consequently, as functions on $\Omega$, $h_e=h_{\langle e\rangle_m}$ for every $1\leq e\leq D-1$, where $\langle e\rangle_m$ is the unique element of $[m]$ congruent to $e$ modulo $m$.

The exponent sums $0$ and $D$ must be kept separate, as shown in Figure~\ref{fig:projective_function_table}(b). On the subgroup, $h_0$ agrees with $h_m$, since both have value one. At the worker $0$, however, $h_0(0)=1$ and $h_m(0)=0$. Similarly, since $D=m+d+1$, the functions $h_D$ and $h_{d+1}$ agree on the subgroup, but $h_D(\infty)=1$ and $h_{d+1}(\infty)=0$.

For the row and column labels above, define
\[
r\boxplus c
=
\begin{cases}
0, & r+c=0,\\
\infty, & r+c=D,\\
\langle r+c\rangle_m, & \text{otherwise}.
\end{cases}
\]
An ordinary label $j\in[m]$ represents $h_j$, while the labels $0$ and $\infty$ represent $h_0$ and $h_D$, respectively. Multiplication of the chosen row and column functions therefore follows the rule $\boxplus$. The resulting function table is shown in Figure~\ref{fig:projective_function_table}.

For each $\ell\in[L]$, the desired entries in column $\ell d$ have labels $\ell d+1,\ldots,\ell d+K$. Taken over all $\ell\in[L]$, the desired labels are every element of $[m]$ except $1,\ldots,d$ and $2d,3d,\ldots,Ld$. The nuisance entries use exactly the labels $0,\infty,1,\ldots,d,2d,3d,\ldots,Ld$, where the final list is empty when $L=1$.

The two additional labels prevent the only collisions between desired and nuisance entries. Under ordinary addition modulo $m$, the mask--mask entry at row $0$ and column $0$ would receive the label $m$, which is the last desired label. Similarly, the entry at row $d$ and column $s$ would receive the label $\langle d+s\rangle_m=\langle D\rangle_m=d+1$, which is the first desired label. The rule $\boxplus$ assigns the new labels $0$ and $\infty$ to these two nuisance entries instead.

Worker $P\in\Omega$ receives
\begin{align*}
X_P
&=
\sum_{k=1}^K f_k(P)A_k
+
f_0(P)R_1
+
f_d(P)R_2,\\
Y_P
&=
\sum_{\ell=1}^L g_{\ell d}(P)B_\ell
+
g_0(P)S_1
+
g_s(P)S_2,
\end{align*}
and returns $Z_P=X_PY_P$. Figure~\ref{fig:projective_function_table}(c) gives these encodings at the subgroup points and at $0$ and $\infty$, using $t^s=t$ for $t\in H$.

\begin{proposition}[Decodability]
\label{prop:projective_decodability}
The projective-line scheme is linearly decodable.
\end{proposition}

\begin{proof}
We first show that $h_0,h_1,\ldots,h_m,h_D$ form a basis of $\Fq^\Omega$. Consider a linear relation among these functions. Evaluating at $0$ shows that the coefficient of $h_0$ is zero, since $h_0$ is the only function in the list that is nonzero there. Evaluating at $\infty$ then shows that the coefficient of $h_D$ is zero.

The remaining relation has the form $\sum_{e=1}^m c_et^e=0$ for every $t\in H$. Dividing by $t$ gives a polynomial of degree at most $m-1$ that vanishes at the $m$ distinct elements of $H$. Hence every coefficient is zero, and $h_1,\ldots,h_m$ are linearly independent.

The list contains $m+2=|\Omega|$ functions, so it forms a basis of the full function space. The desired and nuisance labels identified above are disjoint parts of this basis. The desired entries are therefore linearly independent modulo $\Ncal$, and Theorem~\ref{thm:decodability_condition} gives linear decodability.
\end{proof}

\begin{proposition}[Privacy]
\label{prop:projective_privacy}
The projective-line scheme is $2$-private.
\end{proposition}

\begin{proof}
At a subgroup worker $(t,1)$, the $A$-side mask-coefficient row is $(1,t^d)$. Since $m=d(L+1)-1$, we have $\gcd(d,m)=1$. The map $t\mapsto t^d$ therefore permutes $H$, so these rows are distinct as $t$ ranges over $H$. Any two of them are consequently linearly independent.

The workers $0$ and $\infty$ give the mask rows $(1,0)$ and $(0,1)$, respectively. Since $t^d$ is nonzero for every $t\in H$, each endpoint row is linearly independent from every subgroup row. Hence any two $A$-side mask rows are linearly independent.

On the $B$-side, the mask-coefficient row at $(t,1)$ is $(1,t^s)=(1,t)$ because $s=m+1$ and $t^m=1$. These rows are distinct as $t$ ranges over $H$, and the workers $0$ and $\infty$ again give $(1,0)$ and $(0,1)$. Hence any two $B$-side mask rows are also linearly independent.

Both mask families satisfy the $2$-MDS mask condition. Corollary~\ref{cor:mds_privacy} therefore gives $2$-privacy.
\end{proof}

\begin{proof}[Proof of Theorem~\ref{thm:projective_line_construction}]
The assumption $KL+K+L\mid(q-1)$ guarantees that the subgroup required by the construction exists. The projective-line scheme uses $KL+K+L+2$ workers. Proposition~\ref{prop:projective_decodability} gives linear decodability, while Proposition~\ref{prop:projective_privacy} gives $2$-privacy.
\end{proof}

\begin{proof}[Proof of Corollary~\ref{cor:t2_two_worker_gap}]
Theorem~\ref{thm:general_lower_bound} gives the lower bound $KL+K+L$, while Theorem~\ref{thm:projective_line_construction} gives the upper bound $KL+K+L+2$ under the stated divisibility condition. Combining the two bounds proves the result.
\end{proof}

\begin{proof}[Proof of Corollary~\ref{cor:projective_field_optimality}]
Suppose that $KL+K+L+1$ is a prime power, and set $q=KL+K+L+1$. Then $KL+K+L=q-1$, so Theorem~\ref{thm:projective_line_construction} gives a scheme with $N=q+1$ workers.

The construction uses $2$-MDS masks, while Proposition~\ref{prop:mds_field_size_bound} gives $N\leq q+1$ for every $T=2$ scheme with MDS masks. The projective-line construction therefore meets this bound with equality.
\end{proof}

\bibliographystyle{IEEEtran}
\bibliography{ref}

\begin{thebibliography}{10}
\providecommand{\url}[1]{#1}
\csname url@samestyle\endcsname
\providecommand{\newblock}{\relax}
\providecommand{\bibinfo}[2]{#2}
\providecommand{\BIBentrySTDinterwordspacing}{\spaceskip=0pt\relax}
\providecommand{\BIBentryALTinterwordstretchfactor}{4}
\providecommand{\BIBentryALTinterwordspacing}{\spaceskip=\fontdimen2\font plus
\BIBentryALTinterwordstretchfactor\fontdimen3\font minus
  \fontdimen4\font\relax}
\providecommand{\BIBforeignlanguage}[2]{{%
\expandafter\ifx\csname l@#1\endcsname\relax
\typeout{** WARNING: IEEEtran.bst: No hyphenation pattern has been}%
\typeout{** loaded for the language `#1'. Using the pattern for}%
\typeout{** the default language instead.}%
\else
\language=\csname l@#1\endcsname
\fi
#2}}
\providecommand{\BIBdecl}{\relax}
\BIBdecl

\bibitem{10415397}
O.~Makkonen and C.~Hollanti, ``General framework for linear secure distributed
  matrix multiplication with byzantine servers,'' \emph{IEEE Transactions on
  Information Theory}, vol.~70, no.~6, pp. 3864--3877, 2024.

\bibitem{9004505}
R.~G.~L. D’Oliveira, S.~El~Rouayheb, and D.~Karpuk, ``Gasp codes for secure
  distributed matrix multiplication,'' \emph{IEEE Transactions on Information
  Theory}, vol.~66, no.~7, pp. 4038--4050, 2020.

\bibitem{9508383}
R.~G.~L. D’Oliveira, S.~El~Rouayheb, D.~Heinlein, and D.~Karpuk, ``Degree
  tables for secure distributed matrix multiplication,'' \emph{IEEE Journal on
  Selected Areas in Information Theory}, vol.~2, no.~3, pp. 907--918, 2021.

\bibitem{11195364}
C.~Hofmeister, R.~Bitar, and A.~Wachter-Zeh, ``Cat and dog: Improved codes for
  private distributed matrix multiplication,'' in \emph{2025 IEEE International
  Symposium on Information Theory (ISIT)}, 2025, pp. 1--6.

\bibitem{10858081}
O.~Makkonen, E.~Saçıkara, and C.~Hollanti, ``Algebraic geometry codes for
  secure distributed matrix multiplication,'' \emph{IEEE Transactions on
  Information Theory}, vol.~71, no.~4, pp. 2373--2382, 2025.

\bibitem{8647313}
W.-T. Chang and R.~Tandon, ``On the capacity of secure distributed matrix
  multiplication,'' in \emph{2018 IEEE Global Communications Conference
  (GLOBECOM)}, 2018, pp. 1--6.

\bibitem{8382305}
H.~Yang and J.~Lee, ``Secure distributed computing with straggling servers
  using polynomial codes,'' \emph{IEEE Transactions on Information Forensics
  and Security}, vol.~14, no.~1, pp. 141--150, 2019.

\bibitem{8675905}
J.~Kakar, S.~Ebadifar, and A.~Sezgin, ``On the capacity and
  straggler-robustness of distributed secure matrix multiplication,''
  \emph{IEEE Access}, vol.~7, pp. 45\,783--45\,799, 2019.

\bibitem{8985291}
M.~Aliasgari, O.~Simeone, and J.~Kliewer, ``Private and secure distributed
  matrix multiplication with flexible communication load,'' \emph{IEEE
  Transactions on Information Forensics and Security}, vol.~15, pp. 2722--2734,
  2020.

\bibitem{9229375}
Q.~Yu and A.~S. Avestimehr, ``Coded computing for resilient, secure, and
  privacy-preserving distributed matrix multiplication,'' \emph{IEEE
  Transactions on Communications}, vol.~69, no.~1, pp. 59--72, 2021.

\bibitem{9539194}
Z.~Jia and S.~A. Jafar, ``On the capacity of secure distributed batch matrix
  multiplication,'' \emph{IEEE Transactions on Information Theory}, vol.~67,
  no.~11, pp. 7420--7437, 2021.

\bibitem{8989342}
W.-T. Chang and R.~Tandon, ``On the upload versus download cost for secure and
  private matrix multiplication,'' in \emph{2019 IEEE Information Theory
  Workshop (ITW)}, 2019, pp. 1--5.

\bibitem{9440909}
J.~Kakar, A.~Khristoforov, S.~Ebadifar, and A.~Sezgin, ``Codes trading upload
  for download cost in secure distributed matrix multiplication,'' \emph{IEEE
  Transactions on Communications}, vol.~69, no.~8, pp. 5409--5424, 2021.

\bibitem{9162296}
R.~G.~L. D’Oliveira, S.~E. Rouayheb, D.~Heinlein, and D.~Karpuk, ``Notes on
  communication and computation in secure distributed matrix multiplication,''
  in \emph{2020 IEEE Conference on Communications and Network Security (CNS)},
  2020, pp. 1--6.

\bibitem{yu2017polynomial}
Q.~Yu, M.~Maddah-Ali, and S.~Avestimehr, ``Polynomial codes: an optimal design
  for high-dimensional coded matrix multiplication,'' \emph{Advances in Neural
  Information Processing Systems}, vol.~30, 2017.

\bibitem{8006963}
K.~Lee, C.~Suh, and K.~Ramchandran, ``High-dimensional coded matrix
  multiplication,'' in \emph{2017 IEEE International Symposium on Information
  Theory (ISIT)}, 2017, pp. 2418--2422.

\bibitem{8437871}
S.~Kiani, N.~Ferdinand, and S.~C. Draper, ``Exploitation of stragglers in coded
  computation,'' in \emph{2018 IEEE International Symposium on Information
  Theory (ISIT)}, 2018, pp. 1988--1992.

\bibitem{8765375}
S.~Dutta, M.~Fahim, F.~Haddadpour, H.~Jeong, V.~Cadambe, and P.~Grover, ``On
  the optimal recovery threshold of coded matrix multiplication,'' \emph{IEEE
  Transactions on Information Theory}, vol.~66, no.~1, pp. 278--301, 2020.

\bibitem{8949560}
Q.~Yu, M.~A. Maddah-Ali, and A.~S. Avestimehr, ``Straggler mitigation in
  distributed matrix multiplication: Fundamental limits and optimal coding,''
  \emph{IEEE Transactions on Information Theory}, vol.~66, no.~3, pp.
  1920--1933, 2020.

\bibitem{yu2019lagrange}
Q.~Yu, S.~Li, N.~Raviv, S.~M.~M. Kalan, M.~Soltanolkotabi, and S.~A.
  Avestimehr, ``Lagrange coded computing: Optimal design for resiliency,
  security, and privacy,'' in \emph{The 22nd International Conference on
  Artificial Intelligence and Statistics}.\hskip 1em plus 0.5em minus
  0.4em\relax PMLR, 2019, pp. 1215--1225.

\bibitem{9519610}
B.~Hasırcıoğlu, J.~Gómez-Vilardebó, and D.~Gündüz, ``Bivariate
  polynomial coding for efficient distributed matrix multiplication,''
  \emph{IEEE Journal on Selected Areas in Information Theory}, vol.~2, no.~3,
  pp. 814--829, 2021.

\bibitem{10786350}
A.~Fidalgo-Díaz and U.~Martínez-Peñas, ``Distributed matrix multiplication
  with straggler tolerance using algebraic function fields,'' \emph{IEEE
  Transactions on Information Theory}, vol.~71, no.~2, pp. 996--1006, 2025.

\bibitem{9965858}
R.~A. Machado and F.~Manganiello, ``Root of unity for secure distributed matrix
  multiplication: Grid partition case,'' in \emph{2022 IEEE Information Theory
  Workshop (ITW)}, 2022, pp. 155--159.

\bibitem{9523544}
J.~Zhu, Q.~Yan, and X.~Tang, ``Improved constructions for secure multi-party
  batch matrix multiplication,'' \emph{IEEE Transactions on Communications},
  vol.~69, no.~11, pp. 7673--7690, 2021.

\bibitem{9681059}
B.~Hasırcıoǧlu, J.~Gómez-Vilardebó, and D.~Gündüz, ``Bivariate
  polynomial codes for secure distributed matrix multiplication,'' \emph{IEEE
  Journal on Selected Areas in Communications}, vol.~40, no.~3, pp. 955--967,
  2022.

\bibitem{6594847}
H.~Randriambololona, ``An upper bound of singleton type for componentwise
  products of linear codes,'' \emph{IEEE Transactions on Information Theory},
  vol.~59, no.~12, pp. 7936--7939, 2013.

\bibitem{7137642}
D.~Mirandola and G.~Zémor, ``Critical pairs for the product singleton bound,''
  \emph{IEEE Transactions on Information Theory}, vol.~61, no.~9, pp.
  4928--4937, 2015.

\bibitem{7133155}
I.~Cascudo, ``Powers of codes and applications to cryptography,'' in \emph{2015
  IEEE Information Theory Workshop (ITW)}, 2015, pp. 1--5.

\bibitem{ball2012sets}
S.~Ball, ``On sets of vectors of a finite vector space in which every subset of
  basis size is a basis.'' \emph{Journal of the European Mathematical Society
  (EMS Publishing)}, vol.~14, no.~3, p. 733, 2012.

\bibitem{9732990}
N.~Mital, C.~Ling, and D.~Gündüz, ``Secure distributed matrix computation
  with discrete fourier transform,'' \emph{IEEE Transactions on Information
  Theory}, vol.~68, no.~7, pp. 4666--4680, 2022.

\bibitem{9965839}
H.~H. López, G.~L. Matthews, and D.~Valvo, ``Secure matdot codes: a secure,
  distributed matrix multiplication scheme,'' in \emph{2022 IEEE Information
  Theory Workshop (ITW)}, 2022, pp. 149--154.

\bibitem{10206764}
R.~A. Machado, G.~L. Matthews, and W.~Santos, ``Hera scheme: Secure distributed
  matrix multiplication via hermitian codes,'' in \emph{2023 IEEE International
  Symposium on Information Theory (ISIT)}, 2023, pp. 1729--1734.

\bibitem{10619357}
O.~Makkonen, ``Flexible field sizes in secure distributed matrix multiplication
  via efficient interference cancellation,'' in \emph{2024 IEEE International
  Symposium on Information Theory (ISIT)}, 2024, pp. 2562--2567.

\bibitem{e25020266}
\BIBentryALTinterwordspacing
E.~Byrne, O.~W. Gnilke, and J.~Kliewer, ``Straggler- and adversary-tolerant
  secure distributed matrix multiplication using polynomial codes,''
  \emph{Entropy}, vol.~25, no.~2, 2023. [Online]. Available:
  \url{https://www.mdpi.com/1099-4300/25/2/266}
\BIBentrySTDinterwordspacing

\bibitem{10478018}
D.~Karpuk and R.~Tajeddine, ``Modular polynomial codes for secure and robust
  distributed matrix multiplication,'' \emph{IEEE Transactions on Information
  Theory}, vol.~70, no.~6, pp. 4396--4413, 2024.

\bibitem{11653890}
C.~Hofmeister, R.~Tajeddine, A.~Wachter-Zeh, and R.~Bitar, ``On the extension
  of private distributed matrix multiplication schemes to the grid partition,''
  in \emph{2026 IEEE International Symposium on Information Theory (ISIT)},
  2026, pp. 1--6.

\bibitem{10619695}
R.~Cartor, R.~G.~L. D'Oliveira, S.~El~Rouayheb, D.~Heinlein, D.~Karpuk, and
  A.~Sprintson, ``Secure distributed matrix multiplication with
  precomputation,'' in \emph{2024 IEEE International Symposium on Information
  Theory (ISIT)}, 2024, pp. 2568--2573.

\bibitem{9606447}
R.~A. Machado, R.~G.~L. D’Oliveira, S.~E. Rouayheb, and D.~Heinlein, ``Field
  trace polynomial codes for secure distributed matrix multiplication,'' in
  \emph{2021 XVII International Symposium "Problems of Redundancy in
  Information and Control Systems" (REDUNDANCY)}, 2021, pp. 188--193.

\bibitem{8832193}
M.~Kim, H.~Yang, and J.~Lee, ``Private coded matrix multiplication,''
  \emph{IEEE Transactions on Information Forensics and Security}, vol.~15, pp.
  1434--1443, 2020.

\bibitem{8754796}
N.~Raviv and D.~A. Karpuk, ``Private polynomial computation from lagrange
  encoding,'' \emph{IEEE Transactions on Information Forensics and Security},
  vol.~15, pp. 553--563, 2020.

\bibitem{9696353}
J.~Li and C.~Hollanti, ``Private and secure distributed matrix multiplication
  schemes for replicated or mds-coded servers,'' \emph{IEEE Transactions on
  Information Forensics and Security}, vol.~17, pp. 659--669, 2022.

\bibitem{9930803}
Y.~Yao, N.~Liu, W.~Kang, and C.~Li, ``Secure distributed matrix multiplication
  under arbitrary collusion pattern,'' \emph{IEEE Transactions on Information
  Forensics and Security}, vol.~18, pp. 85--100, 2023.

\bibitem{10161614}
O.~Makkonen and C.~Hollanti, ``Secure distributed gram matrix multiplication,''
  in \emph{2023 IEEE Information Theory Workshop (ITW)}, 2023, pp. 192--197.

\bibitem{nomeir2025quantum}
M.~Nomeir, A.~Aytekin, L.~Hu, and S.~Ulukus, ``Quantum private distributed
  matrix multiplication with degree tables,'' \emph{arXiv preprint
  arXiv:2511.23406}, 2025.

\bibitem{11462243}
------, ``Quantum gasp codes for private distributed matrix multiplication,''
  in \emph{ICASSP 2026 - 2026 IEEE International Conference on Acoustics,
  Speech and Signal Processing (ICASSP)}, 2026, pp. 21\,336--21\,340.

\bibitem{11505934}
O.~Makkonen and C.~Hollanti, ``Analog secure distributed matrix
  multiplication,'' \emph{IEEE Transactions on Information Theory}, vol.~72,
  no.~7, pp. 4751--4765, 2026.

\bibitem{ALON_1999}
N.~Alon, ``Combinatorial nullstellensatz,'' \emph{Combinatorics, Probability
  and Computing}, vol.~8, no. 1–2, p. 7–29, 1999.

\end{thebibliography}
\end{document}